\documentclass[reqno]{amsart}
\usepackage{amsmath,amssymb,color}
\usepackage{amsfonts, amscd, epsfig, amsmath, amssymb,enumerate}
\usepackage{graphicx}

\newtheorem{theorem}{Theorem}[section]
\newtheorem{lemma}[theorem]{Lemma}
\newtheorem{proposition}[theorem]{Proposition}

\numberwithin{equation}{section}

\newcommand{\mc}[1]{{\mathcal #1}}
\newcommand{\mf}[1]{{\mathfrak #1}}
\newcommand{\mb}[1]{{\mathbf #1}}
\newcommand{\bb}[1]{{\mathbb #1}}

\renewcommand{\epsilon}{\varepsilon}

\newcommand{\RR}{\mathbb R}

\newcommand{\bM}{\mathbf{M}}

\let\ve=\varepsilon

\let\ve=\varepsilon

\begin{document}

\author{C\'edric Bernardin}

\address{\noindent
Universit\'e de Lyon and CNRS, UMPA, UMR-CNRS 5669, ENS-Lyon,
46, all\'ee d'Italie, 69364 Lyon Cedex 07 - France.
\newline e-mail: \rm \texttt{Cedric.Bernardin@umpa.ens-lyon.fr}
}

\author{Claudio Landim} 

\address{\noindent IMPA, Estrada Dona Castorina 110, CEP 22460 Rio de
  Janeiro, Brasil and CNRS UMR 6085, Universit\'e de Rouen, Avenue de
  l'Universit\'e, BP.12, Technop\^ole du Madril\-let, F76801
  Saint-\'Etienne-du-Rouvray, France.  \newline e-mail: \rm
  \texttt{landim@impa.br} }

\title[Entropy of nonequilibrium stationary measures ]{ Entropy of stationary nonequilibrium  measures of boundary driven symmetric simple exclusion processes}

\noindent

\begin{abstract}
  We examine the entropy of stationary nonequilibrium measures
  of boundary driven symmetric simple exclusion processes. In
  contrast with the Gibbs--Shannon entropy \cite{B, DLS2}, the entropy of nonequilibrium stationary states differs from the entropy of local equilibrium states.
\end{abstract}

\maketitle
\thispagestyle{empty}

\section{Introduction}

In the last decade important progress has been accomplished in the
understanding of nonequilibrium stationary states through the study of
stochastic lattice gases (\cite{BDGJL9, D} and references therein).

The simplest nontrivial example of such dynamics is the
one-dimensional simple symmetric exclusion process on the finite
lattice $\{1, 2, \dots , N-1\}$ with particle reservoirs coupled to
the sites $1$ and $N-1$.  In this model the microscopic states are
described by the vector $\eta = (\eta(1), \eta(2), \dots, \eta(N-1))$,
where $\eta(i)=1$ if the site $i$ is occupied and $\eta(i)=0$ if the
site is empty.  Each particle, independently from the others, perform
a nearest-neighbor symmetric random walk with the convention that each
time a particle attempts to jump to a site already occupied the jump
is suppressed. At the boundaries, particles are created and destroyed
in order for the density to be $\alpha$ at the left boundary and
$\beta$ at the right boundary, $0\le \alpha$, $\beta\le 1$.

We denote by $\mu^N_{\alpha, \beta}$ the stationary state of this
system which is a probability measure in the space of configurations
and which can be expressed in terms of a product of matrices
\cite{DLS}. Since the particle number is the only conserved quantity
in the bulk, in the scaling limit $N\to\infty$, $i/N\to x \in [0,1]$,
the system is described by a single density field $\rho(x)$, $x \in
(0, 1)$. The typical density profile $\bar\rho (x)$ is the stationary
solution of a partial differential equation with boundary conditions.
In the context of symmetric exclusion processes,
\begin{equation*}
\bar\rho (x) \;=\; \alpha (1-x) \;+\; \beta x\;.
\end{equation*}

The nonequilibrium stationary states exhibit long range correlations
\cite{S} which are responsible in the large deviations regime for the
non locality of the free energy functional \cite{DLS, BDGJL2}.  More
precisely, if $\gamma$ stands for a density profile different from the
typical one $\bar\rho$, the asymptotic probability of $\gamma$ is
exponentially small and given by
\begin{equation*}
  \mu^N_{\alpha, \beta} [\gamma(\cdot)]  \sim e^{-N V_{\alpha,
      \beta}(\gamma)}\; ,
\end{equation*}
where the so called nonequilibrium free energy $V_{\alpha, \beta}$ is
a non local functional.

Since in equilibrium the probability of such large deviations is
determined by the induced change in the entropy, it is natural to
investigate the entropy of nonequilibrium stationary states.

Denote by $\mf S_N (\nu^N)$ the Gibbs--Shannon entropy of a state
$\nu^N$:
\begin{equation*}
\mf S_N (\nu^N) \;=\; - \, \sum_{\eta}  \nu^N (\eta) \,  
\log \nu^N (\eta)\;,
\end{equation*}
where the sum is carried over all lattice configurations $\eta$.
Recently, Bahadoran \cite{B} proved that for a large class of
stochastic lattice gases the Gibbs--Shannon entropy of nonequilibrium
stationary states has the same asymptotic behavior as the
Gibbs--Shannon entropy of local equilibrium states. In our context of
boundary driven symmetric simple exclusion processes this result can
be stated as follows. Denote by $\nu^N_{\alpha, \beta}$ the product
measure
\begin{equation*}
  \nu^N_{\alpha, \beta}(\eta) \;=\; \prod_{i=1}^{N-1}
  \bar\rho(i/N)^{\eta(i)} [1-\bar\rho(i/N)]^{1-\eta(i)}\; .
\end{equation*}
Thus, at site $i$, independently from the other sites, we place a
particle with probability $\bar\rho(i/N)$ and leave the site empty
with probability $1-\bar\rho(i/N)$. Bahadoran proved that 
\begin{equation*}
\lim_{N\to \infty} \frac 1N \mf S_N (\mu^N_{\alpha, \beta}) \;=\;
\lim_{N\to \infty} \frac 1N \mf S_N (\nu^N_{\alpha, \beta})\;.
\end{equation*}
The long range correlations of the nonequilibrium stationary state is
therefore not captured by the Gibbs--Shannon entropy. 

Derrida, Lebowitz and Speer \cite{DLS2} showed that for the symmetric
simple exclusion process the difference
\begin{equation*}
\mf S_N (\mu^N_{\alpha, \beta}) \;-\;
\mf S_N (\nu^N_{\alpha, \beta})
\end{equation*}
converges as $N\to\infty$, and that the limit depends on the two
points correlation functions. Hence, the long range correlations
appear in the first order correction to the Gibbs--Shannon entropy.

In this article we examine the entropy of the stationary nonequilibrium states $\mu^N_{\alpha, \beta}$.  In the classical Boltzmann--Gibbs theory of equilibrium statistical mechanics \cite{P},
the steady state $\mu^N_\beta (\eta)$ of a microstate $\eta$ is given
by
\begin{equation}
\label{00}
\mu^N_{\beta} (\eta) = \cfrac{1}{Z_{N}(\beta)} \exp(-\beta H(\eta))
\end{equation}
where $\beta$ is the inverse of the temperature, $H(\eta)$ the energy
of $\eta$ and $Z_N (\beta)$ the partition function.  The Boltzmann
entropy is then defined as the limit, when the degrees of freedom $N$
of the system converges to infinity, of $1/N$ times the logarithm of
the number of microstates with a prescribed energy:
\begin{equation*}
S (E) \;=\; \lim_{\delta \to 0} \lim_{N \to \infty} N^{-1} \log 
\Big( \sum_{\eta} {\bf 1} \{|H(\eta) -N E| \le \delta N\}
\Big) \;,
\end{equation*}
where the summation is performed over all configurations $\eta$ and
where $\mb 1\{A\}$ is the indicator of the set $A$. The pressure $P(\beta)$ is defined by 
\begin{equation*}
P(\beta)=\lim_{N \to \infty} \cfrac{1}{N\beta} \log Z_N (\beta)
\end{equation*}
and the Boltzmann entropy is related to the pressure function by 
\begin{equation*}
S(E)= \inf_{\beta>0} \left\{ \beta P (\beta) +\beta E \right\}\;.
\end{equation*}

In view of \eqref{00} and by analogy, we define the energy of a
microstate $\eta$ as $- \log \mu^N_{\alpha, \beta} (\eta)$ and the entropy of the stationary nonequilibrium measure $\mu^N_{\alpha, \beta}$ by
\begin{equation*}
S_{\alpha, \beta}(E) \;=\; \lim_{\delta \to 0} \lim_{N \to \infty} N^{-1} \log 
\Big( \sum_{\eta \in \Omega_N} {\bf 1} \{| N^{-1} \log
\mu^N_{\alpha, \beta} (\eta)  + E| \le \delta \}
\Big) \;.
\end{equation*} 

We propose in \eqref{eq:varS} a variational formula for the entropy function $S_{\alpha, \beta}$ in terms of the nonequilibrium free energy
$V_{\alpha, \beta}$ and the equilibrium Gibbs--Shannon entropy, that we conjecture to
be valid for a large class of boundary driven stochastic lattice
gases. This formula is based on a strong form of local equilibrium,
stated as assumption ({\bf H}).  We present in \eqref{eq:S1} an
explicit formula for the entropy function $S_{\alpha, \beta}$ and we
show in \eqref{S2} that it is strictly concave, being the
Legendre transform of a strictly concave function $P_{\alpha, \beta}$,
identified as the nonequilibrium pressure. This last point is proved
in section \ref{sec03}.

In Section \ref{seczr} we compute the entropy of stationary nonequilibrium measures of
boundary driven zero range processes and in Section \ref{sec:le} we
show that the entropy of the nonequilibrium stationary
states $\mu^N_{\alpha, \beta}$ is different from the entropy
of the local equilibrium states $\nu^N_{\alpha, \beta}$. In Section
\ref{sec04}, we determine the energy band and describe the density
profiles with lowest and largest energy. In Section \ref{sec:isen}, we
examine the isentropic surfaces and in the appendix we show that the
strong form of local equilibrium holds for the symmetric simple
exclusion process by using the ideas of \cite{DLS}.

\section{Stationary nonequilibrium entropy function}
\label{sec:entropy}

Fix an integer $N\geq 1$, $0<\alpha \leq \beta < 1$ and let $\Lambda_N
:=\{1, \dots , N-1\}$. Denote by $\Omega_N:=\{0,1\}^{\Lambda_N}$ the
configuration space and by $\eta$ the elements of $\Omega_N$, so that
$\eta(x)=1$, resp.\ $0$, if site $x$ is occupied, resp.\ empty, for
the configuration $\eta$.  We denote by $\sigma^{x,y}\eta$ the
configuration obtained from $\eta$ by exchanging the occupation
variables $\eta (x)$ and $\eta (y)$, i.e.\
\begin{equation*}
(\sigma^{x,y} \eta) (z) := 
  \begin{cases}
        \eta (y) & \textrm{ if \ } z=x\\
        \eta (x) & \textrm{ if \ } z=y\\
        \eta (z) & \textrm{ if \ } z\neq x,y,
  \end{cases}
\end{equation*}
and by $\sigma^{x}\eta$ the configuration obtained from
$\eta$ by flipping the configuration at $x$, i.e.\
\begin{equation*}
(\sigma^{x} \eta) (z) := 
  \begin{cases}
        1-\eta (x) & \textrm{ if \ } z=x\\
        \eta (z) & \textrm{ if \ } z\neq x.
  \end{cases}
\end{equation*}

The one-dimensional boundary driven symmetric exclusion process is the
Markov process on $\Omega_N$ whose generator $L_N$ can be decomposed
as
\begin{equation*}
L_N \;=\; L_{0,N} \;+\; L_{-,N} \;+\; L_{+,N} \;,
\end{equation*}
where the generators $L_{0,N}$, $L_{-,N}$, $L_{+,N}$ act on functions
$f:\Omega_N\to \bb R$ as
\begin{eqnarray*}
&& (L_{0,N} f)(\eta) \;=\; \sum_{x=1}^{N-2} 
\big[ f(\sigma^{x,x+1} \eta)-f(\eta)\big] \;, \\
&& \quad
(L_{-,N} f)(\eta) \;=\;  
\big\{ \alpha [1-\eta(1)] + (1-\alpha) \eta(1) \big\}
\, \big[ f(\sigma^{1} \eta)-f(\eta)\big] \\
&&  \qquad (L_{+,N} f)(\eta) \;=\;  
\big\{ \beta [1-\eta(N-1)] + (1-\beta) \eta(N-1) \big\}
\big[ f(\sigma^{N-1} \eta)-f(\eta)\big] \;.
\end{eqnarray*}

We denote by $\eta_t$ the Markov process on $\Omega_N$ with generator
$L_N$.  Since the Markov process $\eta_t$ is irreducible, for each
$N\ge 1$, and $0<\alpha \leq \beta<1$ there exists a unique stationary
state denoted by $\mu^N_{\alpha,\beta}$.

The entropy function $S_{\alpha,\beta}: \bb R_+ \to \{- \infty\}
\cup[0, \log 2]$ associated to the nonequilibrium stationary state
$\mu^N_{\alpha,\beta}$ is defined by
\begin{equation*}
S_{\alpha,\beta} (E) \;=\; \lim_{\delta \to 0} \lim_{N \to \infty}
\cfrac{1}{N} \log \sum_{ \eta \in \Omega_N} {\bf 1} \left\{
\left| N^{-1} \log \mu^N_{\alpha, \beta} (\eta) +E \right| \le
        \delta\right\}
\end{equation*} 
whenever the limits exist. To keep notation simple, we sometimes
denote $S_{\alpha,\beta}$ by $S$.

Note that we may include in the sum $\mu^N_{\alpha,\beta}(\eta)$:
\begin{equation}
\label{l2}
S_{\alpha,\beta} (E) \;=\; E\;+\; \lim_{\delta \to 0} \lim_{N \to \infty}
\cfrac{1}{N} \log \mu^N_{\alpha,\beta} \big \{
\left| N^{-1} \log \mu^N_{\alpha, \beta} (\eta) +E \right| \le
        \delta\big\}\; .
\end{equation} 
In particular, 
\begin{equation}
\label{eq:J}
J_{\alpha,\beta} (E) \;=\; E \;-\; S_{\alpha,\beta}(E)
\end{equation}
is the large deviations rate function of the random variables $-
N^{-1} \log \mu^N_{\alpha,\beta} (\eta)$ under the probability measure
$\mu_{\alpha,\beta}^N$. 

At equilibrium $\alpha=\beta$, the stationary state
$\mu^N_{\alpha,\beta}$ is a Bernoulli product measure with density
$\alpha$ and the entropy function is given by
\begin{equation}
\label{f12}
S_{\alpha} (E) \;:=\;S_{\alpha,\alpha} (E) \;=\;
-\, s\Big(-\cfrac{E+\log(1-\alpha)}{\log \alpha - \log(1-\alpha)} 
\, \Big)\;, 
\end{equation}
where 
\begin{equation*}
s(\theta)= \theta \log \theta + (1-\theta) \log(1-\theta)
\end{equation*}
represents the Gibbs--Shannon entropy. This formula is valid for $E$
in the energy band $[E_{-} (\alpha), E_{+} (\alpha)]$ where
\begin{equation*}
E_{-}(\alpha) \;=\; -\log \big\{ \max (\alpha,1-\alpha)\big\}\;, 
\quad E_{+}(\alpha) \;=\; - \log \big\{ \min (\alpha,1-\alpha)
\big\}\; . 
\end{equation*} 
In the case $\alpha=1/2$ the energy band is reduced to the point $\log
2$ and $S_{1/2} (\log 2)=\log 2$. Outside the energy band we have
$S_{\alpha} (E) = -\infty$.

Identity \eqref{f12} can be derived from the large deviations
principle for the random variable $- N^{-1} \log \mu^N_{\alpha,\alpha}
(\eta)$, which in the equilibrium case is an average of i.i.d. random
variables.

Denote by $\langle \,\cdot\,,\,\cdot\, \rangle$ the scalar product in
${L}^2 ([0,1])$. Let ${\mc M}$ be the set of measurable profiles $m:
[0,1] \to [0,1]$ equipped with the topology induced by weak
convergence, namely $m_n \to m$ in ${\mc M}$ if and only $\langle m_n
, G \rangle \to \langle m, G \rangle$ for every continuous function
$G:[0,1] \to \RR$. For every $m \in {\mc M}$ the nonequilibrium
free energy \cite{DLS, BDGJL2} $V_{\alpha,\beta}(m)$ of $m$ is defined by
$$ 
V_{\alpha,\beta}(m) \;=\; \int_{0}^{1} \Big\{ m(x)
\log{\cfrac{m(x)}{F(x)}}+(1-m(x))\log{\cfrac{1-m(x)}{1-F(x)}}
  +\log{\cfrac{F' (x)}{\beta-\alpha}}\Big\} \, dx\;,
$$
where $F \in C^1 ([0,1])$ is the unique increasing solution of the non
linear boundary value problem
\begin{equation}
\label{eq:bvp}
\begin{cases}
F'' = (m-F)\cfrac{(F')^2}{F(1-F)} \;,\\
F(0)=\alpha\; , \;\;F(1)=\beta \;.
\end{cases}
\end{equation}
To keep notation simple we frequently denote $V_{\alpha,\beta}$ by
$V$.

Decompose the set $\Lambda_N$ into $r=\ve^{-1}$ adjacent intervals
$K_1, \ldots,K_{\varepsilon^{-1}}$ of size $\ve N$ and denote by
$\bM=(M_1, \ldots, M_{r})$ the number of particles in
each box.  Let
\begin{equation*}
\nu (M_1, \ldots, M_r) =\sum_{\eta\in \Omega_N} {\bf 1}
\Big \{ \sum_{x \in K_1} \eta (x) = M_1, \ldots,  
\sum_{x \in K_r} \eta (x) = M_r \Big\} \, \mu^N_{\alpha, \beta} (\eta) 
\end{equation*}
be the probability to find $M_j$ particles in the interval $K_j$,
$1\le j\le r$.  Denote by $\mu^N_{\alpha, \beta}(\cdot| \bM)$ the
probability measure $\mu^N_{\alpha, \beta}$ conditioned to have $M_j$
particles in $K_j$, $j=1, \ldots,r$. The set of configurations $\eta$
such that $\sum_{x \in K_j} \eta (x) = M_j$, $1\le j\le r$, is denoted
by $\Omega_N (\bM)$ and its cardinality by $Z_N (\bM)$.  We shall
assume that for every $0<\alpha \leq \beta<1$, 
\begin{equation*}
\tag*{({\bf H})}
\lim_{\ve \to 0} \limsup_{N \to \infty} \sup_{\bM} 
\sup_{\eta \in \Omega_N (\bM)} \cfrac{1}{N} \, \Big|
\log \big\{ Z_{N} (\bM) \, \mu^N_{\alpha, \beta} (\eta|\bM)\big\} 
\,\Big|\;=\; 0\; .
\end{equation*}
We present in the appendix a formal derivation of this hypothesis.

Assumption ({\bf H}) states that the stationary state $\mu^N_{\alpha,
  \beta}$ conditioned on the number of particles on macroscopic
intervals is uniformly close in a logarithmic sense to the uniform
measure as the number of intervals increases. As we shall see, this
alternative formulation of local equilibrium plays a central role in
the investigation of the entropy of stationary nonequilibrium measures. The first main result of
this article provides a variational formula for the entropy function.
We claim that for every $0<\alpha \leq \beta<1$, $E\ge 0$,
\begin{equation}
\label{eq:varS}
S_{\alpha,\beta} (E) = \sup_{m \in {\mc M}} \big\{ \bb S(m) :
V_{\alpha,\beta}(m) + {\bb S}(m) = E \big\}\;,
\end{equation}
where
\begin{equation*}
{\bb S}(m) \;=\; - \int_{0}^{1}  s (m(x))\, dx\;.
\end{equation*}
\medskip

This formula is a straightforward consequence of assumption ({\bf H})
and the large deviations for the nonequilibrium stationary state
$\mu^N_{\alpha, \beta}$. Indeed, we may rewrite $\mu^N_{\alpha, \beta}
(\eta)$ as $\mu^N_{\alpha, \beta}(\eta|\bM) \nu (\bM)$.  Hence, by
definition of the entropy and by assumption ({\bf H}),
\begin{equation*}
S (E) \;=\; \lim_{\delta \to 0} 
\lim_{\ve \to 0} \lim_{N \to \infty}
\cfrac{1}{N} \log \sum_{\eta \in \Omega_{N}} {\bf 1}
\Big \{ \big | N^{-1} \log \nu (\bM) - N^{-1} \log {Z}_N (\bM) 
+ E \big| \le \delta \Big\}\;.
\end{equation*}
The previous sum can be rewritten as
\begin{equation*}
\sum_{\bM} Z_N (\bM) \, {\bf 1} \Big\{ \big| N^{-1} \log \nu (\bM) 
- N^{-1} \log {Z}_N (\bM) + E \big|\le \delta\Big\}\;.
\end{equation*}
Recall that $N^{-1} \log \nu(\bM) \sim -V(m)$ and that
$N^{-1} \log Z_N (\bM) \sim {\bb S}(m)$ where $m(\cdot)$ is the
macroscopic profile associated to $\bM$:
\begin{equation*}
m \;=\; \sum_{i=1}^{r} \rho_i \, {\bf 1}\{[x_i, x_{i+1})\}\;, \quad 
\rho_i = M_i /(N\ve) \;, \quad
K_i \; =\; \{[Nx_{i}], \ldots, [N x_{i+1}] -1\} \; .
\end{equation*}
Since for a fixed $\ve$ the sum over $\bM$ has only a polynomial
number of terms in $N$ and since $Z_{N} (\bM)$ is exponentially large
in $N$, only the term which maximizes $Z_N (\bM)$ contributes.  The
result follows.

Note that by \cite[Theorem 4.1]{BSGJL}, the functional $V_{\alpha,
  \beta} + \bb S$ is continuous in $\mc M$.

\subsection{The nonequilibrium pressure}

Let $A: \RR_* \to \bb R_+$, $P : \RR \to \bb R$ be given by
\begin{equation}
\label{f01}
\begin{split}
& A(\theta)\; =\; A_{\alpha,\beta}(\theta) \;:=\; \int_{\alpha}^{\beta}
\cfrac{dx}{\left[ x^{\theta} +(1-x)^\theta \right]^{1/\theta}}\;,
\quad \theta\not = 0\;,
\\
& \quad P(\theta) \;= \; P_{\alpha,\beta} (\theta) \;:=\;
\theta \, \log \Big( \cfrac{A_{\alpha,\beta} (\theta)}{\beta
  -\alpha}\Big)\; , \quad \theta\not = 0\;,
\end{split}
\end{equation}
$P (0)=-\log 2$.  As we shall see in \eqref{S2}, $P$ is the Legendre
transform of the entropy function $S$ and may thus be identified with
the nonequilibrium pressure. An elementary computation shows that
$\log A$ is strictly increasing on the intervals $(-\infty, 0)$ and
$(0,\infty)$ and that $\lim_{\theta \to \pm 0} \log A(\theta) = \mp
\infty$.  Moreover,
\begin{equation}
\label{f04}
P^{\prime} (\theta) \; =\;  \cfrac{1}{A (\theta)\, \theta^2} 
\int_{(\alpha /(1-\alpha))^\theta}
^{(\beta/ (1-\beta))^\theta} \frac 1{1+x} 
\frac{(1+x)\log (1+x) - x \log x }
{(x^{1/\theta} +1) (x^{\theta -1}+x^{\theta})^{1/\theta}} 
\, dx\;+\;\log \frac{A(\theta)}{\beta -\alpha}
\end{equation}
for $\theta\not =0$, and 
\begin{equation*}
P' (0) \;=\; \log \Big( \cfrac{1}{\beta -\alpha} 
\int_{\alpha}^{\beta} \cfrac{dx}{\sqrt{x (1-x)}} \, \Big) \;.
\end{equation*}

We prove in Section \ref{sec03} the following properties.

\begin{lemma}
\label{s01}
The function $P$ is a $C^2$ strictly concave function. Moreover,
\begin{equation*}
\lim_{\theta \to \pm \infty} \big\{P (\theta) - \theta E_{\mp} \big\}
\;=\; 0\;, \quad
\lim_{\theta \to \pm \infty} \big\{  P (\theta) -\theta
{P}^{\prime} (\theta) \big\} \;=\; 0 \;, 
\end{equation*}
where
\begin{equation*}
\begin{split}
& E_{+} \;=\; E_{+} (\alpha, \beta) 
\;=\; \log \Big( \cfrac{1}{\beta -\alpha} 
\int_{\alpha}^{\beta} \cfrac{dx}{\min \{x, 1-x\}} \, \Big)\; ,
\\
&\quad E_{-} \;=\; E_{-} (\alpha,\beta) \;=\; 
\log\Big(\cfrac{1}{\beta-\alpha} 
\int_{\alpha}^\beta \cfrac{dx}{\max \{x,1-x\}} \, \Big)\; .
\end{split}
\end{equation*}
\end{lemma}

It follows from this lemma that $\lim_{\theta \to \pm \infty} \theta\{
P'(\theta) - E_{\mp}\} =0$, and that for each $E \in (E_{-} ,
E_{+})$, there exists a unique $\theta_E = \theta(\alpha,\beta, E)\in
\bb R$ such that
\begin{equation}
\label{f05}
P^{\prime} (\theta_E) \; =\; E\; .
\end{equation}

Define the functions $\gamma_\pm : \bb R \to \bb R_+$ by
\begin{equation*}
\begin{split}
& \gamma_- (\theta) \;=\; \gamma_- (\alpha, \beta, \theta) \; =\; 
\min \Big\{ \Big( \frac \alpha {1-\alpha} \Big)^\theta \,,\,
\Big( \frac \beta {1-\beta} \Big)^\theta \Big\} \; , \\
&\quad
\gamma_+ (\theta) \;=\; \gamma_+ (\alpha, \beta, \theta) \; =\; 
\max \Big\{ \Big( \frac \alpha {1-\alpha} \Big)^\theta \,,\,
\Big( \frac \beta {1-\beta} \Big)^\theta \Big\} \; ,
\end{split}
\end{equation*}
and let $W_E=W_{\alpha,\beta, E}: [\gamma_- (\theta_E), \gamma_+
(\theta_E)] \to [0,1]$ be the monotone function given by
\begin{equation*}
W_E (x) \;=\; \frac 1{A(\theta_E)\, \theta_E} 
\int_{[\alpha / (1-\alpha)]^{\theta_E}}^{x} 
\cfrac{dt}{ (t^{1/\theta_E} +1) (t^{\theta_E -1} + 
t^{\theta_E})^{1/\theta_E}}\; \cdot
\end{equation*}
Clearly, $W_E ([\alpha/ (1-\alpha)]^{\theta_E}) =0$. On the other
hand, the change of variables $t = [x/ (1-x)]^{\theta_E}$ shows that
$W_E ([\beta/ (1-\beta)]^{\theta_E}) =1$ in view of the definition of
$A(\theta)$.  Let $h_E: [0,1] \to [\gamma_- (\theta_E), \gamma_+
(\theta_E)]$ be the inverse of $W_E$ so that $h_E(0) = (\alpha
/(1-\alpha))^{\theta_E}$, $h_E(1) = (\beta/(1-\beta))^{\theta_E}$.

\subsection{An explicit formula for $S_{\alpha , \beta}$.}

We are now in a position to present an explicit formula for the entropy function $S$. We claim that for every $0<\alpha \leq
\beta<1$, $E_{-}(\alpha , \beta) < E < E_{+} (\alpha , \beta)$,
\begin{equation}
\label{eq:S1}
S_{\alpha , \beta} (E) \;=\; {\bb S} 
\Big( \frac{h_E(x)}{1+h_E(x)} \, \Big) \;.
\end{equation}

Indeed, consider the variational problem (\ref{eq:varS}).  Let $\theta$
be the Lagrange multiplier and let $R(m,\theta)$ be the function
defined by
\begin{equation*}
R(m,\theta) \;=\; {\bb S}(m) \;-\; \theta \, 
\big\{ V(m) + {\bb S}(m) - E \big\}\;.
\end{equation*}
Since by \cite{DLS, BDGJL2} $(\delta V/\delta m) = \log
[m/(1-m)] - \log [F/(1-F)]$, the conditions $(\delta R/\delta m) =
\partial_{\theta} R =0$ imply that
\begin{equation}
\label{eq:mF}
\begin{split}
& \qquad\qquad\qquad\qquad\qquad\qquad
m=\frac{(F/[1-F])^\theta}{1+(F/[1-F])^\theta} \;,\\
& \int_0^1 \Big\{ m(x) \log F(x) + [1-m(x)] \log [1-F (x)] 
-\log\Big( \frac{F' (x)}{\beta-\alpha}\Big)  \Big\} \, dx \;=\; - E\; , 
\end{split}
\end{equation}
where $F$ is the unique increasing solution of the non linear boundary
value problem \eqref{eq:bvp}. We report the first identity in
(\ref{eq:mF}) to (\ref{eq:bvp}) to get that
\begin{equation*}
\cfrac{F''}{F'}=\left\{ 
\cfrac{F^{\theta -1} -(1-F)^{\theta -1}}{(1-F)^{\theta} 
+F^{\theta}}\right\}F'\;.
\end{equation*} 
Since $d/dz \left[ \theta^{-1} \log ((1-z)^\theta + z^\theta)\right] =
\cfrac{z^{\theta -1} -(1-z)^{\theta -1}}{(1-z)^{\theta} +z^{\theta}}$
and $(\log F') '= F'' /F'$ we deduce from the previous equation that
\begin{equation}
\label{f02}
F' \;=\; A \left[ (1-F)^{\theta} + F^{\theta} \right]^{1/\theta}
\end{equation}
for some positive constant $A$ determined by the boundary conditions
satisfied by $F$:
\begin{equation*}
A\;=\; \int_0^1 \frac{F'(x)}
{[ (1-F)^{\theta} + F^{\theta}]^{1/\theta}} \, dx\;.
\end{equation*}
The change of variables $y = F(x)$ shows that $A = A(\theta)$ is given
by \eqref{f01}.

Recall the definition of $\gamma_\pm (\theta)$. Let $g_\theta : [0,1]
\to [\gamma_- (\theta) , \gamma_+ (\theta)]$ be given by $g_\theta =
\left( F/(1-F)\right)^{\theta}$ and observe that
\begin{equation*}
g_\theta' \;=\; A(\theta) \, \theta\, (g_\theta^{1/\theta} +1) \,
(g_\theta^{\theta -1}+g_\theta^{\theta})^{1/\theta}\; .
\end{equation*}
Define $U_\theta : [\gamma_- (\theta) , \gamma_+ (\theta)] \to [0,1]$
by
\begin{equation*}
U_\theta (x) \;=\; \frac 1{A(\theta)\, \theta} 
\int_{(\alpha / (1-\alpha))^{\theta}}^{x} 
\cfrac{dt}{ (t^{1/\theta} +1) (t^{\theta -1} + t^\theta)^{1/\theta}}
\end{equation*}
and remark that $g_\theta = U_\theta^{-1}$. 

In the second equation of (\ref{eq:mF}) replacing $m$ by $g_\theta
/(1+ g_\theta)$ and $F'$ by the right hand side of identity
\eqref{f02}, we obtain that
\begin{equation*}
\cfrac{1}{\theta} \int_{0}^1 \Big\{ 
\cfrac{g_\theta(x)}{1+g_\theta(x)} \log g_\theta(x) -
\log (1+g_\theta(x))\Big\}\, dx  
\;-\; \log \frac{A(\theta)}{\beta -\alpha} \;=\; - E  \; .
\end{equation*}
Performing the change of variables $y = g_\theta(x)$, we get that
\begin{equation}
\label{eq:theta2b}
\cfrac{1}{A (\theta)\, \theta^2} \int_{(\alpha /(1-\alpha))^\theta}
^{(\beta/ (1-\beta))^\theta} \frac 1{1+x} 
\frac{x \log x - (1+x)\log (1+x)}
{(x^{1/\theta} +1) (x^{\theta -1}+x^{\theta})^{1/\theta}} \, dx
\;- \;\log \frac {A(\theta)}{\beta -\alpha}  \;=\;- E \;.
\end{equation} 
In view of the explicit expression for $P'$, we may rewrite the
previous identity as $P'(\theta) = E$. Therefore, by
\eqref{f05}, $\theta=\theta_E$, and hence $U_\theta = W_E$,
$g_\theta = h_E$. Moreover, in view of \eqref{eq:mF}, the density
profile $m$ which solves the variational problem \eqref{eq:varS} is
$m= h_E/[1+h_E]$. This proves \eqref{eq:S1}. \medskip

\subsection{A variational formula for $S_{\alpha , \beta}$.}

We conclude this section showing that $P_{\alpha , \beta}$ is the
Legendre transform of $S_{\alpha , \beta}$ and can therefore be
identified with the nonequilibrium pressure. 

For every $0<\alpha \leq \beta<1$, $E\ge 0$,
\begin{equation}
\label{S2}
S (E) \;=\; \inf_{\theta \in \RR} 
\big \{ \theta E -P (\theta) \big \}\;.
\end{equation}
If $E$ belongs to the energy band $(E_{-} , E_{+})$ the infimum is
attained at $\theta_E$ given by \eqref{f05} and
\begin{equation*}
S (E) \;=\; \theta_E E \;-\; P (\theta_E)\;.
\end{equation*}
Moreover, $S(E_\pm)=0$ and $S (E) = -\infty$ if $E \notin [E_{-} ,
E_{+} ]$.

By abuse of notation we shall call $S$ the Legendre transform of $P$.
Usually the Legendre transform is defined as a supremum and involves
convex functions. However, by taking a minus sign we may transform
convex functions into concave functions and supremums into infimums.

The proof of \eqref{S2} is simple. In section \ref{sec04} we show that
$S(E)=-\infty$ outside $[E_- , E_+]$ and that $S(E_\pm)=0$.  On the
other hand, by Lemma \ref{s01}, $\theta E -P (\theta)$ is a monotone
non-decreasing function for $E \ge E_+$. Hence, for $E \ge E_+$,
$\inf_{\theta \in \RR} \{ \theta E -P (\theta) \} = \lim_{\theta\to
  -\infty} \{ \theta E -P (\theta) \}$. By Lemma \ref{s01} again,
$\{\theta E -P (\theta)\}$ converges to $-\infty$, $0$ for $E>E_+$,
$E=E_+$, respectively. Therefore, by the first observation of the
proof, $S (E) =\inf \{ \theta E - P (\theta)\}$ for $E\ge E_+$.  The
case $E \le E_{-}$ is analogous.

By Lemma \ref{s01}, $\lim_{\theta \to \pm \infty} \{\theta E
-P(\theta)\} =+\infty$ for $E \in (E_- , E_+)$ and the function
$\theta \to \theta E -P(\theta)$ is strictly convex on $\RR$.  Hence,
$\inf \{ \theta E - P (\theta)\} = \theta_E\, E - P (\theta_E)$, where
$\theta_E$ solves \eqref{f05}. We may rewrite this expression as
$\theta_E\, P'(\theta_E) - P(\theta_E)$.  In view of \eqref{f04}, to
conclude the proof of \eqref{S2} it remains to show that the first
term on the right hand side of \eqref{f04} multiplied by $\theta$ and
computed at $\theta = \theta_E$ coincides with (\ref{eq:S1}). This can
be shown by performing the change of variables $u=h_{E} (x)$ in
\eqref{eq:S1} and recalling that $W_E$ is the inverse of $h_E$.

It follows from \eqref{S2} that $S$ is concave and that $P$, the
Legendre transform of the entropy, can be identified with the
pressure.

The equilibrium case can be recovered by letting $\alpha \to
\beta$. In this case,
\begin{equation*}
\begin{split}
& \theta(\alpha,\alpha, E) \;=\; - \, \cfrac{1}
{\log [\alpha/(1-\alpha)]} \, \cfrac{E+\log(1-\alpha)}
{E+\log(\alpha)}\, \; , \\
& \quad h_E (x) \;=\; - \, \cfrac{E+\log(1-\alpha)}
{E+\log(\alpha)}\;\cdot
\end{split}
\end{equation*} 

\section{Boundary driven zero range processes}
\label{seczr}

We compute in this section the entropy of stationary nonequilibrium measures of boundary driven
zero range processes. The model is described by a positive integer
variable $\eta(x)$ representing the number of particles at site $x\in
\Lambda_N$.  The state space $\bb N^{\Lambda_N}$ is denoted
$\Omega_N$. At exponential times one particle jumps with rate
$g(\eta(x))$ to one of the nearest--neighbor sites.  The function
$g:\bb N\to \bb R_+$ is increasing and $g(0)=0$.  We assume that the
system interacts with particle reservoirs at the boundary of
$\Lambda_N$ whose activity at the right is $\varphi_+>0$ and at the
left is $\varphi_->0$.  The microscopic dynamics is defined by the
generator
\begin{equation*}
\mc L_N = \mc L_{0,N} + \mc L_{-,N} + \mc L_{+,N}
\end{equation*}
where
\begin{equation*}
\begin{split}
& \mc L_{0,N} f(\eta)  = \sum_{x=1}^{N-2}
\Big\{ g(\eta(x))\, [ f(T^{x,x+1}\eta) -f (\eta)]
\;+\; g(\eta(x+1)) \, [ f(T^{x+1,x}\eta) -f (\eta)] \Big\} \\
& \quad \mc L_{-,N} f(\eta) \;=\;  
g(\eta(1)) \, [ f(S^-_1 \eta) -f (\eta)]
\;+\; \varphi_- \, [f(S^+_1\eta) - f(\eta)]   \;, \\
& \qquad \mc L_{+,N} f(\eta) \;=\; 
g(\eta(N-1)) \, [ f(S^-_{N-1} \eta) - f (\eta)]
\;+\; \varphi_+ \, [f(S^+_{N-1}\eta) - f(\eta)] \;,
\end{split}
\end{equation*}
in which 
\begin{equation*}
(T^{x,y}\eta)(z) = \left\{
\begin{array}{ccl}
\eta(z) &\hbox{if}& z\neq x,y \\
\eta(z) -1 &\hbox{if}& z=x \\
\eta(z) +1  &\hbox{if}& z=y 
\end{array}
\right.
\end{equation*}
is the configuration obtained from $\eta$ when a
particle jumps from $x$ to $y$, and
\begin{equation*}
(S^\pm_x \eta)(z) = 
\left\{
\begin{array}{ccl}
\eta(z) &\hbox{if}& z\neq x \\
\eta (z) \pm 1 &\hbox{if}& z=x 
\end{array}
\right.
\end{equation*}
is the configuration where we added (resp. subtracted) one particle at
$x$. Note that, since $g(0)=0$, the number of particles cannot become
negative. 

The invariant measures of the boundary driven zero range processes can
be computed explicitly. Let $\varphi_N:\Lambda_N \to \bb R_+$ be the
linear interpolation between $\varphi_-$ and $\varphi_+$: 
\begin{equation}
\label{zr02}
\varphi_N(x) \;=\; \Big(1- \frac xN \Big) \, \varphi_- 
\;+\; \frac xN\, \varphi_+\;.
\end{equation}
The invariant measure $m^N_{\varphi_-, \varphi_+}$ is the product
measure whose  marginals are given by
\begin{equation*}
m^N_{\varphi_-, \varphi_+} \{\eta : \eta(x) = k\}
\; = \;\frac {1}{Z(\varphi_N(x))} \; 
{\frac {\varphi_N(x)^k}{g(1)\cdots g(k)}} \;, \quad k\ge 0\;,
\end{equation*}
where $Z(\varphi) = 1 + \sum_{k\ge 1} \varphi^k / [g(1) \cdots g(k)]$
is the normalization constant. 

Denote by $R:\bb R_+\to \bb R_+$ the density of particle under the
stationary state with activity equal to $\varphi$ on both boundaries:
\begin{equation*}
R(\varphi) \;=\; E_{m^N_{\varphi, \varphi}} [\eta(x)] \;=\;
\frac 1{Z(\varphi)} \sum_{k\ge 1} k\, \frac{\varphi^k}{g(1) \cdots
  g(k)}\; =\; \frac{\varphi Z'(\varphi)}{Z(\varphi)} \; ,
\end{equation*}
and by $\Phi:\bb R_+\to \bb R_+$ the inverse of $R$:
\begin{equation*}
  \Phi \;=\; R^{-1}\;.
\end{equation*}

Under the stationary state, the typical density profile
$\bar\rho:[0,1]\to \bb R_+$ is the unique solution of the elliptic
equation
\begin{equation*}
\left\{
\begin{aligned}
&  \Delta \Phi(\bar\rho) = 0 \;,\\
& \bar\rho(0) = R(\varphi_-)\; , \; \bar\rho(1) = R(\varphi_+)\;,
\end{aligned}
\right.
\end{equation*}
where $\Delta$ stands for the Laplacian. As $N\uparrow\infty$, the
activity profile $\varphi_N$ introduced in \eqref{zr02} converges to
$\Phi(\bar\rho)$:
\begin{equation*}
\lim_{N\to\infty} \hat\varphi_N \;=\; \Phi(\bar\rho)\;,
\end{equation*}
where $\hat\varphi_N :[0,1]\to \bb R_+$ is the function defined by
$\hat\varphi_N (0) =\varphi_-$, $\hat\varphi_N (x/N) =\varphi_N(x)$,
$x\in\Lambda_N$, $\hat\varphi_N (1) =\varphi_+$ and extended to the
interval $[0,1]$ by linear interpolation.

The weight of a configuration $\eta$ under the stationary state
$m^N_{\varphi_-, \varphi_+}$ is given by
\begin{equation*}
m^N_{\varphi_-, \varphi_+} (\eta) \;=\;  \exp  
\sum_{x=1}^{N-1} \Big\{  \eta(x) \log \varphi_N(x) \,-\,
\log [g(1)\cdots g(\eta(x))] \,-\, \log Z(\varphi_N(x))  \Big\}\, .
\end{equation*}
In the special case where $g(k)= \mb 1 \{k\ge 1\}$, the weight
$m^N_{\varphi_-, \varphi_+} (\eta)$ is a function of the empirical
density. Hence, if in this case we define for a profile
$\rho:[0,1]\to\bb R_+$,
\begin{equation}
\label{zr03}
\mc H_{\varphi_-, \varphi_+} (\rho) \;:=\; 
\int_0^1 \rho(x) \log \Phi(\bar\rho(x)) \;-\; 
\log Z(\Phi(\bar\rho(x))) \; dx\;,
\end{equation}
we have that
\begin{equation*}
m^N_{\varphi_-, \varphi_+} (\rho) \;\sim\; e^{N \mc H_{\varphi_-,
    \varphi_+} (\rho)} \;.
\end{equation*}

In general, the weight of a configuration is not a function of the
empirical density but a function of the field associated to the
variables $\xi(x) = \eta(x) - \log [g(1)\cdots g(\eta(x))]/\log
\varphi_N(x)$.

The nonequilibrium free energy functional is easy to compute in the
context of zero range boundary driven systems since the stationary
state is a product measure. A simple computation shows that
\begin{equation*}
V_{\varphi_-, \varphi_+} (\rho) \;=\; \int_0^1 \rho(x) \log 
\frac{\Phi(\rho(x))}{\Phi(\bar\rho(x))} \;-\; \log
\frac{Z(\Phi(\rho(x)))}{Z(\Phi(\bar\rho(x)))} \, dx\; .
\end{equation*}

To present an explicit formula for the entropy function in this
context, we need to introduce some notation borrowed from the theory
of large deviations of i.i.d. random variables.  Fix $\varphi >0$,
let $M:\bb R\to \bb R_+$ be given by
\begin{equation*}
M_\varphi(a) \;=\; \frac 1{Z(\varphi)}
\sum_{k\ge 0}  \frac{\varphi^k}{g(1) \cdots g(k)} e^{a F_\varphi(k)}\; ,
\end{equation*}
where $F_\varphi(k) = k - [\log \varphi]^{-1} \log [g(1) \cdots
g(k)]$, and let $R_\varphi(a) = M_\varphi'(a)/M_\varphi(a)$. The large
deviations rate function $I_\varphi : \bb R \to \bb R_+$ for the mean
of the i.i.d. random variables $\xi(x) = \eta(x) - [\log \varphi]^{-1}
\log [g(1) \cdots g(\eta(x))]$, $1\le x\le N-1$, distributed according
to $m^N_{\varphi , \varphi}$ is given by
\begin{equation*}
I_\varphi (x) \;=\; x R_\varphi^{-1}(x) - \log M_\varphi(R_\varphi^{-1}(x))\;.
\end{equation*}

In the particular case where $g(k) = \mb 1 \{k\ge 1\}$, we get that
$Z(\varphi)=(1-\varphi)^{-1}$, $R(\varphi) = \varphi /(1-\varphi)$,
$\Phi(\rho) = \rho/(1+\rho)$, $F_\varphi(k) = k$, $\xi (x) = \eta(x)$,
$M_\varphi(a) = (1-\varphi)/(1- \varphi e^a)$, $R_\varphi(a) = \varphi
e^{a}/[1-\varphi e^{a}]$ so that
\begin{equation}
\label{zr01}
\begin{split}
I_\varphi (x) \; &=\; x \log \frac x \varphi 
- (1+x) \log (1+x) - \log (1-x) \\
&=\; x \log \frac{\Phi(x)}{\varphi} - \log
\frac{Z(\Phi(x))}{Z(\varphi)}\;\cdot
\end{split}
\end{equation}
We emphasize that formulas \eqref{zr03} and \eqref{zr01} have been
deduced only in the case $g(k) = \mb 1\{k\ge 1\}$, and may not hold in
general.

For each $0<\varphi_- < \varphi_+$, define the entropy function
$S_{\varphi_-, \varphi_+} : \bb R \to \bb R$ by
\begin{equation*}
S_{\varphi_-, \varphi_+} (E) \;=\; \lim_{\delta \to 0} \lim_{N \to \infty}
\cfrac{1}{N} \log \sum_{ \eta \in \Omega_N} {\bf 1} \left\{
\left| N^{-1} \log m^N_{\varphi_-, \varphi_+} (\eta) +E \right| \le
        \delta\right\}
\end{equation*} 
whenever the limits exist. We may introduce in the sum
$m^N_{\varphi_-, \varphi_+} (\eta)$ to get that the entropy function
is equal to
\begin{equation*}
E \;+ \; \lim_{\delta \to 0} \lim_{N \to \infty}
\cfrac{1}{N} \log  m^N_{\varphi_-, \varphi_+} \left\{
\left| N^{-1} \log m^N_{\varphi_-, \varphi_+} (\eta) +E \right| \le
        \delta\right\}
\end{equation*}
Since $\log m^N_{\varphi_-, \varphi_+} (\eta)$ may be expressed in
terms of the variables $\{\xi(x) : 1\le x\le N-1\}$, which are
independent under $m^N_{\varphi_-, \varphi_+}$, the large deviations
principle gives that
\begin{equation*}
S_{\varphi_-, \varphi_+} (E) \;-\; E \;=\; -\, \inf_\lambda  \int_0^1 
I_{\Phi(\bar\rho(x))} (\lambda(x)) \, dx  \;, 
\end{equation*}
where the infimum is carried over all profiles $\lambda :[0,1]\to\bb
R$ such that
\begin{equation*}
\int_0^1 \lambda(x) \log \Phi(\bar\rho(x)) - \log Z(\Phi(\bar\rho(x))) 
\, dx \;=\;  -E
\end{equation*}

In view of \eqref{zr01}, \eqref{zr03}, in the case where $g(k)= \mb
1\{k\ge 1\}$, the entropy function becomes
\begin{equation}
\label{l1}
S_{\varphi_-, \varphi_+} (E) \;-\; E \;=\; -\, \inf_\rho  \int_0^1 
\rho(x) \log \frac{\Phi(\rho(x))}{\Phi(\bar\rho(x))} \;-\; \log
\frac{Z(\Phi(\rho(x)))}{Z(\Phi(\bar\rho(x)))} \, dx \;, 
\end{equation}
where the infimum is carried over all density profiles $\rho:
[0,1]\to\bb R_+$ such that $\mc H_{\varphi_-, \varphi_+} (\rho)= - E$.
Therefore, in the case $g(k)= \mb 1\{k\ge 1\}$, where an explicit
formula is available, up to a linear term, $S_{\varphi_-, \varphi_+}
(E)$ is obtained by minimizing the free energy functional
$V_{\varphi_-, \varphi_+}$ over all density profiles $\rho$ with
energy $\mc H_{\varphi_-, \varphi_+} (\rho)$ equal to $- E$.

Finally, if we define $\widehat S_{\varphi_-, \varphi_+} : \bb R \to\bb R$ by
\begin{equation*}
\widehat S(E) \;=\; -\, \inf \Big\{ \int_0^1 \rho \log
\frac{\Phi(\rho)}{\Phi(\bar\rho)} - \log \frac{Z(\Phi(\rho))}
{Z(\Phi(\bar\rho))} \, dx : \int_0^1 [\rho - \bar\rho] \log 
\Phi(\bar\rho) \, dx = - E \Big\}\;,
\end{equation*}
we obtain that
\begin{equation*}
S_{\varphi_-, \varphi_+}(E) \;=\;  E \;+\; \widehat S_{\varphi_-,
  \varphi_+} \Big( E + \int_0^1 \Big \{\bar \rho \log
\Phi(\bar\rho) - \log Z(\Phi(\bar\rho))\Big\} \, dx \Big)\;.
\end{equation*}
Note that $\widehat S(E)\le 0$ and $\widehat S(0)= 0$. As above, we
stress that the identity \eqref{l1} and all formulas thereafter were
derived in the case $g(k) = \mb 1\{k\ge 1\}$.

\section{The nonequilibrium pressure}
\label{sec03}

We prove in this section Lemma \ref{s01}. Recall the definition of the
function $P$ introduced in \eqref{f01}. We first prove that
$P$ is strictly concave. A long and tedious computation concluded
with the change of variables $t = x^\theta/[x^\theta + (1-x)^\theta]$
shows that for $\theta\not = 0$,
\begin{equation}
\label{eq:PP}
\begin{split}
P''(\theta) \; &=\; \frac 1{\theta^3} 
\int_A^B s(t)^2 \, \mu_\theta(dt) \; -\; \frac 1{\theta^3} 
\Big( \int_A^B s(t) \, \mu_\theta(dt) \Big)^2 \\
& - \; \frac 1{\theta^2} 
\, \int_A^B t(1-t) \, \Big( \log \frac t{1-t} \Big)^2 
\, \mu_\theta(dt)\; ,  
\end{split}
\end{equation}
where $A= \alpha^\theta/[\alpha^\theta + (1-\alpha)^\theta]$, $B=
\beta^\theta/[\beta^\theta + (1-\beta)^\theta]$,
\begin{equation*}
\mu_\theta(dt) \;:=\; \frac{1}{Z(\theta)} m_\theta(t) \, dt \;,
\quad m_\theta(t) \;:=\; \frac 1{t(1-t)}
\, \frac 1{\frac 1{t^{1/\theta}} + \frac 1{(1-t)^{1/\theta}} }\; ,
\end{equation*}
and $Z(\theta)$ is a normalizing constant which makes $\mu_\theta$ a
probability measure on $[A,B]$.

By Schwarz inequality, the first line of the expression of $P''$
without $\theta^{-3}$ is positive. Therefore, $P$ is strictly
concave on the interval $(-\infty, 0)$. The strict concavity on the
interval $(0,\infty)$ follows from the claim that for all $\theta>0$,
$0<\alpha < \beta<1$,
\begin{equation}
\label{ff01}
\int_A^B s(t)^2 \, \mu_\theta(dt) \;-\;
\Big( \int_A^B s(t) \, \mu_\theta(dt) \Big)^2 \;<\;
\theta\, \int_A^B t(1-t) \, \Big( \log \frac t{1-t} \Big)^2 
\, \mu_\theta(dt)\; ,  
\end{equation}

It is enough to prove that
\begin{equation*}
\int \int_{A\le r\le t \le B} \{ s(t) - s(r) \}^2 \, 
\mu_\theta(dr) \mu_\theta(dt) \;<\; \theta\, 
\int_A^B t(1-t) \, \Big( \log \frac t{1-t} \Big)^2 \, \mu_\theta(dt)\;.
\end{equation*}
Let $H(u) = \theta^{-1} \{u^{-1/\theta} + (1-u)^{-1/\theta}\}$. Assume
that $\theta \not =1$ and denote by $R$ the primitive of $H$ given by
$R(u)=(\theta-1)^{-1} \{ u^{1-1/\theta} - (1-u)^{1-1/\theta}\}$.
Hence, by Schwarz inequality, the left hand side of the previous
inequality is bounded above by
\begin{equation*}
\begin{split}
& \int \int_{A\le r\le t \le B} \Big( \int_r^t \{ s'(u) \}^2 \, H(u)^{-1}
du\Big) \, \Big( \int_r^t \, H(v) dv \Big)\, 
\mu_\theta(dr) \mu_\theta(dt) \\
& \quad =\;
\frac 1{Z^2} \int_A^B du\, \{ s'(u) \}^2 \, H(u)^{-1}
\int_A^u dr \int_u^B dt \, m_\theta(r) \, m_\theta(t) 
\, \{R(t) - R(r) \}\;. 
\end{split}
\end{equation*}
Since $s'(r) = \log[r/(1-r)]$, to conclude the proof of Claim
\eqref{ff01} for $\theta\not =1$, it remains to show that
\begin{equation}
\label{ff02}
\frac 1{Z} \int_A^u dr \int_u^B dt \, m_\theta(r) \, m_\theta(t) 
\, \{R(t) - R(r) \} \; < \; 1
\end{equation}
for all $\theta>0$, $\theta \not = 1$ and $0<A \le u\le B<1$.

The left hand side of the previous inequality can be written as
\begin{equation*}
\int_u^B m_\theta(t) \, R(t)\, dt \; \frac 1{Z} \int_A^u m_\theta(r)
\, dr \;-\; \int_A^u m_\theta(r) \, R(r) \, dr \; \frac 1{Z}
\int_u^B m_\theta(t) \, dt\;.
\end{equation*}
We need to show that this expression is strictly bounded above by $1$
for $A\le u\le B$. Let $K$ be a primitive of $m_\theta \, R$ and
rewrite the previous expression as
\begin{equation*}
\begin{split}
J_{A,B}(u) \; :& =\; [ K(B) - K(u) ]\, M(u)  \;-\; [K(u) - K(A)] \, 
[1-M(u)] \\
& =\; K(B) \, M(u) \;+\; K(A)\, [1-M(u)] \;-\; K(u)\; ,
\end{split}
\end{equation*}
where $M(u) = Z^{-1} \int_{[A,u]} m_\theta(r) \, dr$. This expression
represents the difference between the convex combination of $K(A)$ and
$K(B)$, with weights $M(u)$, $1-M(u)$, and $K(u)$. For $A\le u\le B$,
this difference is clearly absolutely bounded by the variation of $K$
on the interval $[A,B]$:
\begin{equation*}
\sup_{A\le u\le B} \big| J_{A,B}(u) \big| \;\le\; 
\sup_{A\le v\le B} K(v) \;-\; \inf_{A\le v\le B} K(v)\;.
\end{equation*}
Maximizing over $0\le A \le B\le 1$, we get that
\begin{equation*}
\sup_{0\le A\le u\le B\le 1}  \big| J_{A,B}(u) \big|
\;\le\; \sup_{0\le v\le 1} K(v) \;-\; \inf_{0\le v\le 1} K(v)\; .
\end{equation*}

A simple computation shows that
\begin{equation*}
K(t)\;=\; - \, \frac{\theta}{\theta-1} \, \log \{ t^{1/\theta} +
(1-t)^{1/\theta} \} \; .
\end{equation*}
In particular, for $\theta \not =1$, $K(0)=K(1)=0$, $K$ is symmetric
around $1/2$, $K'(1/2) =0$, $K$ decreases on the interval $[0,1/2]$
and increases on the interval $[1/2,1]$. Its total variation on the
interval $[0,1]$ is therefore equal to $-K(1/2) = [\theta/(\theta-1)]
\log 2^{1-1/\theta} = \log 2 < 1$. This proves \eqref{ff02} and
therefore Claim \eqref{ff01} for $\theta\not =1$. The proof for
$\theta=1$ is identical, the only difference being the explicit
expression for the primitives.

The behavior of $P$ in a neighborhood of $0$ is obtained through
a simple Taylor expansion of the integrand. We have
\begin{equation*}
P (\theta) \; =\; -\log 2 \;+\; \theta \log I_1
\;-\; {\theta}^2 \, \cfrac{I_2}{I_1} \; +\; O (\theta^3) \; ,
\end{equation*}
where 
\begin{equation*}
I_1 \;=\; \cfrac{1}{\beta -\alpha} \int_\alpha^\beta
\cfrac{dx}{\sqrt{x(1-x)}} 
\end{equation*}
and 
\begin{equation*}
I_2 \; =\; \cfrac{1}{2(\beta -\alpha)} \int_\alpha^\beta 
\cfrac{[\log x]^2 + [\log(1-x)]^2} {\sqrt{x(1-x)}} \, dx\; .
\end{equation*}
This completes the proof of the strict concavity of $P$ on
${\mathbb R}$.

We now turn to the claim that
\begin{equation*}
\lim_{\theta \to \pm \infty} \big\{ P  (\theta) 
- \theta E_{\mp} \big\} \;=\; 0\;.  
\end{equation*}
We consider the limit $\theta\uparrow\infty$, the other one being
similar. By definition of $P$, we have to prove that
\begin{equation}
\label{f03}
\lim_{\theta \to \infty} \theta \big\{ \log \Big( \frac{A
  (\theta)}{\beta - \alpha} \Big) - E_{-} \big\} \;=\; 0 \;.  
\end{equation}
A preliminary computation shows that $\lim_{\theta \to \infty} \log [
A (\theta)/(\beta - \alpha)] = E_-$.

The proof of \eqref{f03} depends on the positions of $\alpha$ and
$\beta$ with respect to $1/2$, the most difficult case being when
$0<\alpha \le 1/2 \le \beta <1$. Write $A (\theta)$ as
\begin{equation*}
\int_{\alpha}^{1/2} \cfrac{dx}{1-x} \exp\Big\{ -\cfrac{1}{\theta}
\log \Big[ 1+ \Big( \frac x{1-x} \Big)^\theta \Big] \Big\} \;+\; 
\int_{1/2}^{\beta} \cfrac{dx}{x} \exp\Big\{ -\cfrac{1}{\theta}
\log \Big[ 1+ \Big( \frac {1-x}x \Big)^\theta \Big] \Big\}\;.
\end{equation*}
We concentrate on the first integral. Since $|e^q -1 -q | \le
(q^2/2) e^{|q|}$, $q\in \RR$, the first integral is equal to
\begin{equation*}
\int_{\alpha}^{1/2} \cfrac{dx}{1-x} \;-\; 
\frac 1\theta \int_{\alpha}^{1/2} \cfrac{dx}{1-x} 
\log \Big[ 1+ \Big( \frac x{1-x} \Big)^\theta \Big] \;+\; 
\frac 1{\theta^2} {\varepsilon}(\theta)\;,
\end{equation*}
where ${\varepsilon}(\theta)$ is a remainder absolutely bounded by
$[\log 2]^2$ for $\theta>1$. The second integral in this expression
vanishes as $\theta\uparrow\infty$ by the dominated convergence
theorem. Hence, $\log [A (\theta)/(\beta - \alpha)] = E_- +
o(\theta^{-1})$, where $\theta o(\theta^{-1})$ vanishes as
$\theta\uparrow\infty$. This proves \eqref{f03}.

We finally consider the last statement of the lemma.  By \eqref{f04}
and by the change of variables $x=u^{\theta}$,
$P(\theta) - \theta P^{\prime} (\theta)$ is equal to
\begin{equation*}
\begin{split}
& \cfrac{1}{\theta A(\theta)}  \int_{(\alpha /(1-\alpha))^\theta}
^{(\beta/ (1-\beta))^\theta} \frac 1{1+x}  \frac{x \log x - (1+x)\log (1+x)}
{(x^{1/\theta} +1) (x^{\theta -1}+x^{\theta})^{1/\theta}} \, dx \\
& \quad = \cfrac{1}{A(\theta)}
\int_{\alpha/{1-\alpha}}^{\beta/{1-\beta}} \cfrac{du}{(1+u)} \, 
\frac{ u^{\theta} \log u^{\theta} - (1+u^{\theta})\log(1+u^{\theta})
  } {(1+u^{\theta})^{1 + [1/\theta]}} \;\cdot
\end{split}
\end{equation*}
We examine the case $\theta\uparrow\infty$, $0< \alpha \le 1/2 \le
\beta$, the other ones being simpler. Since $A(\theta)$ converges to a
constant as $\theta\uparrow\infty$, only the integral has to be
estimated. By the dominated convergence theorem,
\begin{equation*}
\lim_{\theta\to \infty} \int_{\alpha/{1-\alpha}}^{1} \cfrac{du}{(1+u)} \, 
\frac{ u^{\theta} \log u^{\theta} - (1+u^{\theta})\log(1+u^{\theta})
  } {(1+u^{\theta})^{1 + [1/\theta]}} \;=\; 0
\end{equation*}
because the numerator vanishes as $\theta\uparrow\infty$. On the other
hand, the integral in the interval $[1,\beta/(1-\beta)]$ can be
written as
\begin{equation*}
-\; \int_1^{\beta/{1-\beta}} \cfrac{du}{u (1+u)} \, 
\frac{ \log(1+u^{-\theta}) + u^{-\theta} \log(1+u^{\theta})} 
{(1+u^{-\theta})^{1 + [1/\theta]}}\;\cdot
\end{equation*}
By the dominated convergence theorem, this expression vanishes as
$\theta\uparrow\infty$. This concludes the proof of the lemma.

\section{Energy band}
\label{sec04}

In this section, we determine the energy band $[E_{-},E_{+}]$, i.e.
the range of $V+{\bb S}$. Recall that ${\mc M}$ is the set of profiles
$m:[0,1] \to [0,1]$ and let ${\mc M}_+$, ${\mc M}_-$ be the set of
profiles $m$ of the form $m(x)={\bf 1} \{[0,x_0]\}$, $m(x)={\bf 1}
\{[x_0 , 1]\}$, respectively, for some $x_0 \in [0,1]$.

% \begin{figure}
% \begin{center}
% \label{fig:Emaxmin}
% \includegraphics[width=13cm]{FIN_EPLUSMINUS.png}
% \caption{The two energy band manifolds $E_{\pm}(\alpha,\beta)$}
% \end{center}
% \end{figure}

\begin{lemma}
\label{s03}
For every $0<\alpha \le \beta <1$,
\begin{equation}
\label{f06}
\sup_{m \in {\mc M}} \big\{ V (m) + {\bb S} (m) \big\}
\;=\; \sup_{m \in {\mc M}_+} \big\{ V (m) + {\bb S} (m) 
\big\} \;=\; E_+ \; ,
\end{equation}
and the supremum is attained for a unique profile $m \in {\mc M}_+$.
\end{lemma}

\begin{proof}
By \cite[(2.11)]{BSGJL},
\begin{equation}
\label{f08}
V(m) + {\bb S} (m) \;=\; \sup_{F \in {\mc F}} {\mc G}(m,F)
\end{equation}
where
\begin{equation*}
{\mc G}(m,F) \;=\; - \, \int_{0}^1  \; \Big\{ m(x) \log F(x)
+ [1-m(x)] \log [1-F(x)] - \log \frac{F^{\prime} (x)}{\beta
    -\alpha} \Big\} \, dx  \;,
\end{equation*}
and ${\mc F}$ is the set of all $C^{1}$ increasing functions $F:[0,1]
\to [0,1]$ with boundary conditions $F(0)=\alpha$, $F(1)=\beta$.
Moreover, the supremum is achieved at the unique solution $F$ of the
boundary value problem (\ref{eq:bvp}).

We claim that for all $F$ in $\mc F$,
\begin{equation}
\label{f07}
\sup_{m \in {\mc M}} {\mc G}(m,F) \;=\;
\sup_{m \in {\mc M}_+}  {\mc G} (m,F) \;.
\end{equation}
The first identity in \eqref{f06} follows from this assertion.  To
check \eqref{f07}, fix $F$ in $\mc F$ and note that the supremum is
achieved by $m={\bf 1}\{[0,x_F]\} \in {\mc M}_+$, where $x_F =\sup \{
x \in [0,1]; F(x) \le 1/2\}$, because $F$ is increasing.  Of course
$x_F = 1$ if $\beta \le 1/2$ and $x_F = 0$ if $\alpha \ge 1/2$.

We claim that the solutions of the variational problem \eqref{f06}
belong to $\mc M_+$. Assume that there exists $m_0 \in {\mc M}$ such
that $ V (m_0) + {\bb S} (m_0) = \sup_{m \in {\mc M}} \{ V (m) + {\bb
  S} (m) \}$.  Let $F_0$ be the unique solution of the boundary value
problem (\ref{eq:bvp}) associated to $m_0$, let $x_0 = \sup \{x \in
[0,1]; F_0 (x) \le 1/2 \}$, let $m_1 = {\bf 1} \{[0,x_0]\}$, and let
$F_1$ be the solution of (\ref{eq:bvp}) associated to $m_1$. By
\eqref{f08},
\begin{equation*}
{\mc G} (m_1 , F_0) \;\le\; {\mc G} (m_1 , F_1) \;=\; 
V (m_1) \;+\; {\bb S} (m_1)\;.
\end{equation*}
By definition of $m_0$ and by \eqref{f08}, the previous expression is
bounded above by
\begin{equation*}
\sup_{ m \in {\mc M}} \{ V(m) + {\bb S} (m) \} \;=\;
{\mc G} (m_0, F_0) \;\le\; \sup_{ m \in \mc M} {\mc G} (m, F_0)
\;=\; {\mc G} (m_1 , F_0)\;,
\end{equation*}
where the last identity follows from the argument presented in the
previous paragraph. Since the first and the last terms in this
sequence of inequalities are the same, all terms are equal and ${\mc
  G} (m_1 , F_0) = {\mc G} (m_1, F_1)$. Since, by \eqref{f08},
$\sup_{F} {\mc G} (m_1 , F)$ is uniquely attained at $F=F_1$, $F_0 =
F_1$. Therefore, in view of (\ref{eq:bvp}), $m_1 =m_0$ a.s.

Let $m_{x_0}= {\bf 1} \{[0,x_0 ]\}$, $x_0 \in [0,1]$, and let $F$ be
the solution of the nonlinear boundary value problem (\ref{eq:bvp})
with $m=m_{x_0}$. On the interval $[0,x_0]$, $F(x)=\alpha e^{a x}$ for
some $a>0$, and on the interval $[x_0,1]$, $F(x)=1-(1-\beta)
e^{A(1-x)}$ for some $A>0$. Since $F$ must belong to $C^1([0,1])$, $a$
and $A$ satisfy
\begin{equation*}
\begin{cases}
\alpha a e^{a x_0} =(1-\beta) A e^{A(1-x_0)}\; ,\\
\alpha e^{a x_0} = 1-(1-\beta) e^{A (1-x_0)}\; .
\end{cases}
\end{equation*}
Thus,
\begin{equation}
\label{f10}
e^{a x_0} \;=\; \cfrac{A}{ a+A}\, \frac 1 \alpha \; , \quad
e^{A(1-x_0)} \;=\; \cfrac{a}{a+A} \, \frac 1{1-\beta} \; ,
\end{equation}
and therefore,
\begin{equation}
\label{f09}
\frac 1a \log \Big\{ \frac A{a+A} \frac 1\alpha \Big\}  
\;+\; \frac 1A \log \Big\{ \frac a{a+A} \frac 1{1-\beta} \Big\} 
\;=\; 1\;.
\end{equation}
Moreover, since $x_0$ belongs to $[0,1]$, $a$, $A$ must satisfy
$\max\{1-\beta, 1-\alpha e^a\} \le a/(a+A) \le \min\{1-\alpha,
(1-\beta) e^A\}$. Let $f$, $g:[0,1]\to \bb R$ be given by $f(x) =
\alpha e^{ax}$, $g(x) = 1 - (1-\beta)e^{A(1-x)}$ and let $h(x) = f(x)
- g(x)$. Since $h$ is convex and $0\le x_0\le 1$, $h(x_0) =
h'(x_0)=0$, $h(0) \ge 0$, $h(1) \ge 0$. Hence, $1-\alpha \le (1-\beta)
e^A$, $1-\alpha e^a \le 1-\beta$ so that
\begin{equation}
\label{f11}
1-\beta \;\le\; \cfrac{a}{a+A} \le 1-\alpha\; \cdot
\end{equation}

By the explicit expression of $F$ and by \eqref{f10}, \eqref{f09},
${\bb S} (m_{x_0}) \; +\; V (m_{x_0})$ is equal to
\begin{equation*}
\begin{split}
& \int_{0}^{1} \log \Big\{ \frac{F' (x)}{\beta -\alpha} \Big\} \, dx 
\;-\;  \int_0^{x_0} \log F(x) \, dx \;-\; \int_{x_0}^1 \log (1-F(x))
\, dx \\
&\qquad =\; \frac 1a \log \Big( \frac A{a+A} \frac 1\alpha \Big) 
\log a \;+\; \frac 1A \log \Big( \frac a{a+A} \frac 1{1-\beta} \Big) 
\log A \;-\; \log (\beta - \alpha)\;.
\end{split}
\end{equation*}
Therefore, by the concavity of the $\log$ function and by \eqref{f10},
\begin{equation*}
\begin{split}
{\bb S} (m_{x_0}) \; +\; V (m_{x_0}) \; &\le\; 
\log \Big\{ \log \Big( \cfrac{A}{A+a} \, \cfrac{1}{\alpha}\Big) 
\;+\; \log \Big( \cfrac{a}{a+A} \, \cfrac{1}{1-\beta} \Big)\Big\}
\;-\; \log(\beta -\alpha)\\
& =\; \log \Big\{ \log \Big[ \cfrac{a}{a+A} 
\Big( 1-\cfrac{a}{a+A}\Big) \cfrac{1}{\alpha \, (1-\beta)}
\Big] \Big\} \;-\; \log(\beta -\alpha) 
\end{split}
\end{equation*}
By \eqref{f11}, the previous expression is bounded by
\begin{equation*}
\sup_{p \in [1-\beta, 1-\alpha]} 
\log \Big\{ \log \Big[ \cfrac{p\, (1-p)}{\alpha\, (1-\beta)}
\Big] \Big\} \;-\; \log(\beta -\alpha) \;=\; E_+ \;, 
\end{equation*}
where the last identity follows by a direct computation. The last
supremum is realized for $p=1-\beta$ if $\beta \le 1/2$, for $p=1/2$
if $\alpha \le 1/2 \le \beta$ and for $p=1-\alpha$ if $1/2 \le
\alpha$. 

Up to this point, we proved that $\sup_{m\in \mc M} \{ \bb S (m) + V
(m)\} \le E_+$. Assume that $\beta\le 1/2$ and set $x_0=1$. In this
case, by \eqref{f10}, $a/(a+A) = 1 -\beta$ and all inequalities in the
previous argument are in fact identities. In particular, $\bb S
(m_{1}) + V (m_{1}) = \sup_{m\in \mc M} \{ \bb S (m) + V (m)\} =
E_+$. Moreover, since $\log$ is a strictly concave function and since
by \eqref{f10} $a/(a+A) > (1-\beta)$ for $x_0<1$, $\bb S (m_{1}) + V
(m_{1}) > \bb S (m_{x_0}) + V (m_{x_0})$ for $x_0<1$.  In the same,
way, if $\alpha\ge 1/2$, $\bb S (m_{0}) + V (m_{0}) = \sup_{m\in \mc
  M} \{ \bb S (m) + V (m)\} = E_+$ and $\bb S (m_{0}) + V (m_{0}) >
\bb S (m_{x_0}) + V (m_{x_0})$ for $x_0>0$.

Finally, if $\alpha \le 1/2 \le \beta$, let $x =
\log(2\alpha)/\log[4\alpha (1-\beta)]$ and observe that
\begin{equation*}
\bb S (m_{x}) + V (m_{x}) \;=\; \sup_{m\in \mc M} 
\{ \bb S (m) + V (m)\} \;=\; E_+
\end{equation*}
and that $\bb S (m_{x}) + V (m_{x}) > \bb S (m_{x_0}) + V (m_{x_0})$
for $x_0\not = x$.
\end{proof}

\begin{lemma}
\label{s04}
For every $0<\alpha \le \beta <1$,
\begin{equation*}
\inf_{m \in {\mc M}} \big\{ V (m) + {\bb S} (m) \big\}
\;=\; \inf_{m \in {\mc M}_-} \big \{ V(m) + {\bb S} (m) \big\}
\;=\; E_{-} \;.
\end{equation*}
\end{lemma}

\begin{proof}
Recall \eqref{f08}. Since for each $F$ in $\mc F$, ${\mc G}(\cdot,F)$
is a continuous function for the weak topology, $V+{\bb S}$ is lower
semicontinuous. In view of the explicit expression of $\mc G$ and by
Jensen's inequality, $V+{\bb S}$ is bounded below. Hence, there exists
$m_0 \in {\mc M}$ such that $V(m_0) + {\bb S} (m_0) = \inf_{m \in {\mc
    M}} \{ V (m) + {\bb S} (m) \}$.

Let $E^*_- = \inf_{m \in {\mc M}} \{ V (m) + {\bb S} (m) \} =
\inf_{m\in \mc M} \sup_{F\in \mc F} \mc G(m,F)$. This expression is
bounded below by $\sup_{F\in \mc F} \inf_{m\in \mc M} \mc G(m,F)$. By
the explicit expression of $\mc G$,
\begin{equation*}
\inf_{m\in \mc M} \mc G(m,F) \;=\; \mc G(m_F,F)\;,
\end{equation*}
where $m_F$ belongs to $\mc M_-$. If $\alpha\ge 1/2$, $\beta\le 1/2$,
$m_F \equiv 1$, $m_F \equiv 0$, respectively. Otherwise, $m_F = \mb
1\{[x_F, 1]\}$, where $x_F$ is the unique point where $F$ is equal to
$1/2$.

Assume that $\alpha \ge 1/2$. In this case, $E^*_- \ge \sup_{F\in \mc F}
\mc G(\mb 1,F) = V (\mb 1) + {\bb S} (\mb 1)  \ge E^*_-$. To
conclude the proof of the lemma it remains to compute $\sup_{F\in \mc F}
\mc G(\mb 1,F)$ which can be done as in the previous lemma. The case
$\beta \le 1/2$ is similar.

Assume that $\alpha < 1/2 < \beta$. In this case we have that $E^*_- \ge
\sup_{F\in \mc F} \mc G(m_F,F) = \sup_{0<x<1} \sup_{F\in \mc F_x} \mc
G(\mb 1\{[x,1]\},F)$, where $\mc F_x = \{F\in \mc F : F(x) = 1/2\}$.
Fix $0<x<1$. In each interval $[0,x]$, $[x,1]$, the variational
problem $\sup_{F\in \mc F_x} \mc G(\mb 1\{[x,1]\},F)$ is similar to
the one in \eqref{f08}. In the interval $[0,x]$ the solution $F_{x,0}$
of this variational problem solves the differential equation
\eqref{eq:bvp} with $m\equiv 0$ and boundary conditions $F(0)= \alpha$,
$F(x)=1/2$. Analogously, in the interval $[x,1]$ the solution
$F_{x,1}$ solves the differential equation \eqref{eq:bvp} with
$m\equiv 1$ and boundary conditions $F(x)=1/2$, $F(1) =\beta$.  These
solutions can be computed explicitly and one obtains that
\begin{equation*} 
\begin{split}
& F_{x,0}(y) \;=\; 1 - (1-\alpha) \exp\Big\{ - \frac{\log [2(1-\alpha)]}x \,
y \Big\} \;, \quad 0\le y\le x\;, \\
& \quad F_{x,1}(y) \;=\;  \beta \exp\Big\{ \frac{\log (2\beta)}{1-x} \, 
(y-1) \Big\}\; , \quad x\le y\le 1\; , \\
&\quad \sup_{F\in \mc F_x} \mc G(\mb 1\{[x,1]\},F) \;=\;
x\log \frac{\log 2 (1-\alpha)} x \;+\; (1-x) \log \frac {\log
  2\beta}{1-x} \;-\; \log (\beta - \alpha) \; .
\end{split}
\end{equation*}
Maximizing over $x$ we deduce that $x^*= \log [2(1-\alpha)] /\log
[4(1-\alpha)\beta]$ is the optimal value of $x$ and that
\begin{equation*}
E^*_- \;\ge\; \log \Big\{ \frac 1{\beta - \alpha} \int_\alpha^\beta
\frac {dx}{\max \{x,1-x\}} \Big\} \; \cdot 
\end{equation*}
Moreover, a simple computation shows that for $F_{x^*,0}'(x^*) =
F_{x^*,1}'(x^*)$. Hence, the function $G:[0,1] \to [\alpha, \beta]$
defined by $G(y) = F_{x^*,0} (y) \mb 1\{y\in [0,x^*]\} + F_{x^*,1} (y)
\mb 1\{y\in (x^*,1]\}$ belongs to $\mc F$ and solves the boundary
value problem \eqref{eq:bvp} for $m= \mb 1\{[x^*,1]\}$. Hence, by
definition of $E^*_-$ and by \eqref{f08}
\begin{equation*}
\begin{split}
& E^*_- \;\le \; \bb S(\mb 1\{[x^*,1]\}) \;+\; V(\mb 1\{[x^*,1]\}) 
\; =\; \mc G(\mb 1\{[x^*,1]\}, G) \\
&\qquad \;=\; 
\log \Big\{ \frac 1{\beta - \alpha} \int_\alpha^\beta
\frac 1{\max \{x,1-x\}} \Big\}  \;. 
\end{split}
\end{equation*}
This proves the lemma and shows that a profile with minimum energy is
given by
\begin{equation*}
m_{\alpha, \beta} \; =\; \mb 1\Big\{ \Big[ \frac {\log [2(1-\alpha)]}
{\log [4(1-\alpha)\beta]} \,,\, 1 \Big] \Big\} \;.
\end{equation*}
\end{proof}

We proved in the previous lemma that
\begin{equation*}
\inf_{m\in \mc M} \, \sup_{F\in \mc F} \, \mc G(m,F) \;=\;
\sup_{F\in \mc F} \,\inf_{m\in \mc M} \, \mc G(m,F) \; .
\end{equation*}

Fix $0<\alpha \le \beta <1$ and assume that $V (m) + {\bb S} (m)=
E^*_-$. Then, 
\begin{equation}
\label{s06}
\text{$m(1-m) =0$  a. s.}
\end{equation}

Indeed, fix a profile $m$ in $\mc M$. Since $V (m) + {\bb S} (m) = \mc
G(m,F)$, where $F$ is the solution of \eqref{eq:bvp},
\begin{equation*}
  \frac{\delta [V + {\bb S} ] (m) }{\delta m} 
  \;=\; \frac{\delta \mc G (m,F)}{\delta m} \;+\;
  \frac{\delta \mc G (m,F)}{\delta F} \, \frac{\delta F}{\delta m}\; 
  \cdot
\end{equation*}
Since $F$ solves the Euler equation \eqref{eq:bvp} associated to the
variational problem $\sup_G \mc G(m,G)$, $\delta G (m,F)/\delta F=0$.
On the other hand, $\delta \mc G / \delta m = - \log [F/(1-F)]$, so
that
\begin{equation*}
\frac{\delta [V + {\bb S} ] (m) }{\delta m}  \;=\; -\, \log 
\frac F {1-F} \; \cdot 
\end{equation*}
This expression does not vanish because $F$ is strictly
increasing. Therefore, the extremal values of $V + {\bb S}$ are
attained at the boundary. 

This formal argument can be made rigorous. By the proof of Theorem 7.1
in \cite{bdgjl13}, $V$ and therefore $V+\bb S$ is G\^ateaux
differentiable, and the G\^ateaux derivative of $V+\bb S$ at $m$ is
equal to $-\log [F/(1-F)]$. \medskip

It follows from the previous results and the variational
formula \eqref{eq:varS} that for $0<\alpha \le \beta <1$,
\begin{equation}
\label{s05}
\text{$S(E)=-\infty$ for $E \notin [E_- , E_+]$ and $S(E_\pm)=0$.}
\end{equation}

\section{Isentropic surface}
\label{sec:isen}

We determine in this section the isentropic surfaces defined by

\begin{equation*}
{\mc E}_K \; =\; \{ E \ge 0 : S_{\alpha,\beta} (E)=K \}, \quad K \in
[0,\log 2]\; . 
\end{equation*}

\subsection{The equilibrium case} 

Assume that $\alpha = \beta$. We have already seen right after
\eqref{f12} that the energy band is reduced to the point $\log 2$ in
the case $\alpha = 1/2$. Assume therefore that $\alpha\not =1/2$ and
fix $K\in [0,\log 2)$. There exist exactly two solutions $0<m_{-} (K)
<1/2 <m_+ (K) <1$ of $-s(m)=K$.  Hence, in view of \eqref{f12}, the
level set ${\mc E}_K = \{E_-(K), E_+(K)\}$, where
\begin{equation*}
E_\pm(K) \;=\; -\, \log (1-\alpha) \; -\; \log \frac \alpha{1-\alpha}
m_{\pm} (K) \; .
\end{equation*}
For $K=\log 2$, ${\mc E}_K$ is the singleton $\{ -(1/2) [\log \alpha +
\log(1-\alpha)]\}$.

\subsection{The nonequilibrium case}

Assume now that $\alpha < \beta$. Let $D: \bb R \to \bb R$ be given by
$D(\theta) = \theta P^{\prime} (\theta) - P (\theta)$.
Since $D^{\prime} (\theta) = \theta P'' (\theta)$ and since, by
Lemma \ref{s01}, $P$ is strictly concave, $D$ is strictly
increasing on $(-\infty;0]$ and strictly decreasing on $[0,+\infty)$.
Moreover $D(0)=\log 2$ and, by Lemma \ref{s01}, $\lim_{\theta \to \pm
  \infty} D(\theta) =0$. In particular, for every $K \in (0,\log2)$
there exist exactly two values $\theta_-(K) < 0 < \theta_+(K) $ such
that $\theta_\pm(K) P^{\prime} (\theta_\pm (K)) - P
(\theta_\pm(K)) = K$.

Fix $0< K < \log 2$. By \eqref{S2}, $E$ belongs to ${\mc E}_K$ if and
only if $K = \theta_E P^{\prime} (\theta_E) - P (\theta_E) =
D(\theta_E)$, where $P'(\theta_E)=E$. Hence $\theta_E = \theta_\pm(K)$
so that $E= P'(\theta_\pm(K))$ and
\begin{equation*}
{\mc E}_K = \big\{ P'(\theta_+(K)), P'(\theta_-(K)) \}\; .
\end{equation*}

If we let $\alpha$ and $\beta$ vary, we see that the $K$-isentropic
surface is composed of the two manifolds $E_K^{-} =
P'(\theta_+(K))$ and $E_K^{+} = P'(\theta_-(K))$ which satisfy
\begin{equation*}
E_K^{-} (\alpha,\beta) \;\le\; E_{K}^{+} (\alpha,\beta) \; .
\end{equation*}

% \begin{figure}
% \label{fig:eplusminus}
% \begin{center}
% \includegraphics[width=13cm]{Final_LevelSet.png}
% \caption{Graph of the isentropic manifolds $E_K^{\pm}(\alpha,\beta)$
%   w.r.t. $(\alpha,\beta)$ for $K=(\log 2)/4$} 
% \end{center}
% \end{figure}

\section{Comparison with local equilibria}
\label{sec:le}

In this section, we compare the entropy function
$S_{\alpha,\beta}$ with the entropy function associated to product
measures with a slowly varying density profile that will be called
local equilibrium entropies.

Let $\nu_{\alpha,\beta}^N$ be the product probability measure on
$\{0,1\}^{N-1}$ given by
\begin{equation*}
\nu_{\alpha,\beta}^N (\eta) \;=\; \prod_{x=1}^{N-1} 
\bar\rho (x/N)^{\eta_x} (1-\bar\rho (x/N))^{1-\eta_x}\; ,
\end{equation*}
where $\bar\rho : [0,1] \to [\alpha, \beta]$ is the stationary profile
$\bar\rho (x) = (1-x) \alpha + x \beta$. Denote by $\tilde S :=
\tilde S_{\alpha,\beta} : \bb R_+ \to \{-\infty\}\cup [0,\log 2]$ the
entropy function corresponding to the Gibbs state
$\nu_{\alpha,\beta}^N$:
\begin{equation}
\label{f13}
\tilde S (E) \;=\; \lim_{\delta \to 0} \lim_{N \to \infty}
\cfrac{1}{N} \log \sum_{ \eta \in \Omega_N} {\bf 1} \left\{
\left| N^{-1} \log \nu^N_{\alpha, \beta} (\eta) +E \right| \le
        \delta\right\}
\end{equation} 
whenever the limits exist.

Let ${\tilde P} :={\tilde P}_{\alpha,\beta} : \bb R \to \bb
R$ be the function defined by
\begin{equation*}
{\tilde P} (\theta) \;=\; 
\cfrac{1}{\beta -\alpha} \int_{\alpha}^{\beta}
\log \Big( \cfrac{1}{ x^\theta +(1-x)^\theta} \Big) \, dx
\end{equation*}

\begin{lemma}
\label{s012}
${\tilde P}$ is a $C^2$ strictly concave function and
\begin{equation*}
\lim_{\theta \to \pm \infty} \frac {\tilde P
(\theta)} \theta \,  \; =\; {\tilde E}_{\mp} \; ,\quad
\lim_{\theta \to \pm \infty} \{ \tilde P  (\theta) 
\,-\, \theta { \tilde P}^{\prime} (\theta) \} \;=\; 0\; ,
\end{equation*}
where
\begin{equation*}
{\tilde E}_-  \;=\; \cfrac{1}{\beta -\alpha} \int_{\alpha}^{\beta} 
\log  \cfrac{1}{\max\{x,1-x\}} \;  dx \; , \;\;
{\tilde E}_+ \;=\; \cfrac{1}{\beta -\alpha} \int_{\alpha}^{\beta} 
\log \cfrac{1}{\min\{x,1-x\}} \; dx \;.
\end{equation*}
Moreover, as $\theta\to 0$,
\begin{equation*}
{\tilde P} (\theta) \;=\; -\log 2 \;+\; 
\cfrac{\theta }{\beta -\alpha} \int_{\alpha}^{\beta} 
\log  \cfrac{1}{\sqrt{x (1-x)}} \; 
dx  \;+\;  O(\theta^2)\; .
\end{equation*}
\end{lemma}

The proof is elementary and left to the reader. It follows from this
result that
\begin{equation*}
\lim_{\theta \to \pm \infty} {\tilde P}^{\prime} (\theta) 
\;=\; {\tilde E}_{\mp} \;.
\end{equation*}

\begin{proposition}
\label{s07}
For $E \ge 0$,
\begin{equation*}
\tilde S (E)\; =\; \inf_{\theta \in \RR} \big \{ \theta E - 
{\tilde P} (\theta) \big\}\;.
\end{equation*}
If $E$ belongs to the energy band $({\tilde E}_-, {\tilde E}_{+})$,
the infimum is attained at ${\tilde \theta}_E = {\tilde \theta}_E
(\alpha,\beta)$ the unique solution of ${\tilde P}^{\prime}
(\theta) =E$ and $\tilde S(E) = {\tilde \theta}_E E -{\tilde P}
({\tilde \theta}_E)$. $\tilde S (E)=-\infty$ if $E \notin [{\tilde E}_-,
{\tilde E}_{+}]$ and $\tilde S ({\tilde E}_{\pm})=0$.
\end{proposition}

\begin{proof}
Multiplying and dividing the indicator in \eqref{f13} by $2^{N-1}$, we
reduce the computation of the entropy to a large deviations problem
for independent Bernoulli random variables and we obtain that
\begin{equation*}
\tilde S (E) \;=\; \sup_{m\in \mc M} \big\{ \bb S(m) : I(m) + \bb S(m) =
E\big\} \;,
\end{equation*}
where $I$ stands for the large deviations rate function given by
\begin{equation*}
I(m)\;=\; \int_0^1 \Big\{ m(x) \log \frac{m(x)}{\bar\rho(x)} \, +\,
[1-m(x)] \log \frac{1-m(x)} {1-\bar\rho(x)} \, \Big\} \, dx \; .
\end{equation*}
One should compare this expression with the variational formula
\eqref{eq:varS} for the nonequilibrium entropy.

Repeating the arguments presented in the proof of \eqref{eq:S1}, we
deduce that
\begin{equation*}
\tilde S (E) \;=\; \bb S \Big(\frac{\bar\rho^\theta}{(1-\bar\rho)^\theta +
  \bar\rho^\theta}  \Big)  \;,
\end{equation*}
where $\theta$ is the unique solution of ${\tilde P}^{\prime}
(\theta) =E$. The rest of the proof is similar to the one of
\eqref{S2}.
\end{proof}

Let $E_0 ={P}^{\prime} (0)$, ${\tilde E}_0 ={\tilde
  P}^{\prime} (0)$. By Lemma \ref{s01} and \ref{s012},
\begin{equation*}
{E}_0 \;=\; \log \Big( \cfrac{1}{\beta - \alpha} 
\int_{\alpha}^{\beta} \cfrac{1}{\sqrt{x(1-x)}}\; dx \Big)\; , \quad 
{\tilde E}_0 \;=\; \cfrac{1}{\beta - \alpha} \int_{\alpha}^{\beta}  
\log \Big( \cfrac{1}{\sqrt{x(1-x)}} \Big)\; dx\; .
\end{equation*} 
By Jensen's inequality, $ {\tilde E}_- < E_{-}$, ${\tilde E}_+ <
E_{+}$ and ${\tilde E}_0 < E_0$. Since $\min \{x,1-x\} \le 1/2 \le
\max \{ x,1-x\}$ and $\sqrt{x(1-x)} \le 1/2$, we may compare all
variables with $\log 2$ to obtain in the end that ${\tilde E}_- <
E_{-} < \log 2 < {\tilde E}_0 < \min\{E_0, {\tilde E}_+\} \le \max
\{E_0, {\tilde E}_+\} < E_+$ in the case $\alpha <\beta$.

\medskip The nonequilibrium and the local equilibrium entropy differ.
For every $0<\alpha<\beta<1$, $S < \tilde S$ in the interval $(E_{-},
{\tilde E}_0)$ and $\tilde S < S$ in the interval $(E_0 , {\tilde
  E}_+)$.

Indeed, fix $E\in (E_{-} , {\tilde E}_0)$.  By Jensen's inequality,
$\theta^{-1} {\tilde P} (\theta) < {\theta}^{-1} P (\theta)$,
$\theta\in\bb R$.  Therefore, for every $\theta>0$, $\theta E -
P(\theta) < \theta E - \tilde P(\theta)$. On the other hand, since
$E<E_0$ and $\theta_{E_0} = 0$, $P'(\theta_E) = E < E_0 =
P'(\theta_{E_0}) = P'(0)$. Hence, since $P'$ decreases,
$\theta_E>0$. A similar argument shows that $\tilde \theta_E>0$.  In
conclusion, by the variational formula for the entropies presented in
\eqref{S2} and \ref{s07},
\begin{equation*}
\begin{split}
S(E) \;& =\; \inf_{\theta\in\bb R} \big\{ \theta E - P(\theta)
  \big\} \;=\; \inf_{\theta>0} \big\{ \theta E - P(\theta)
  \big\} \\
&<\; \inf_{\theta>0} \big\{ \theta E - \tilde P(\theta)
  \big\} \;=\; \inf_{\theta\in\bb R} \big\{ \theta E - \tilde P(\theta)
  \big\} \;=\; \tilde S(E)\; .  
\end{split}
\end{equation*}
A similar argument proves the other claim.\\

Consider the sequence of random variables $Y_N (\eta) = - N^{-1} \log
\mu_{\alpha,\beta}^N (\eta)$, ${\tilde Y}_N (\eta)= - N^{-1} \log
\nu_{\alpha,\beta}^N (\eta)$ defined on the probability space
$\Omega_N$ equipped with the probability measure
$\mu_{\alpha,\beta}^N$, ${\nu}_{\alpha,\beta}^N$, respectively. By
\eqref{l2}, the sequence $(Y_N: N\ge 1)$ satisfies a large deviations
principle with convex rate function $J (E)=E-S(E)$. By similar
reasons, the sequence $({\tilde Y}_N : N\ge 1)$ satisfies a large
deviations principle with convex rate function ${\tilde
  J}(E)=E-{\tilde S} (E)$.

Bahadoran \cite{B} proved that $E_{\mu_{\alpha,\beta}^N} [Y_N]$ and
$E_{\nu_{\alpha,\beta}^N} [{\tilde Y}_N]$ have the same limit given by
the Gibbs-Shannon entropy
\begin{equation*}
{\bar E}=-\int_0^1 s( {\bar \rho} (x) ) \,dx\;,
\end{equation*}
where $ {\bar \rho} (x)=\alpha + (\beta-\alpha)x $. This result can
be recovered from ours.

\begin{lemma}
  The nonnegative rate functions $J$ and ${\tilde J}$ vanish at the same
  and unique point
$$
{\bar E}= P'(1)={\tilde P}' (1)\;.
$$ 
In particular, $Y_N$ under $\mu_{\alpha,\beta}^N$, and ${\tilde Y}_N$
under $\nu_{\alpha,\beta}^N$ converge in probability to $\bar E$.
\end{lemma}

\begin{proof}
By Lemma \ref{s01}, the variational formula \eqref{S2} and the
assertions following this formula, the nonnegative rate function $J$
is strictly convex on $[E_{-},E_+]$, differentiable in $(E_-, E_+)$,
and equal to $+\infty$ outside of the interval $[E_{-},E_+]$.  It has
therefore a unique minimum ${\bar E} \in [E_-,E_+]$. 

By \eqref{S2} and \eqref{f05}, $S'(E)=\theta_E$ on $(E_-, E_+)$, where
$\theta_E$ is the unique solution of $P'(\theta)=E$.  By Lemma
\ref{s01}, $\lim_{\theta \to \pm \infty} P'(\theta)=E_{\mp}$.  It
follows from the previous two facts that $\lim_{E \to E_{\pm}}
J'(E)=\pm \infty$.  Since $J$ is strictly convex, $J$ has a unique
minimizer ${\bar E}$ in $(E_{-}, E_{+})$ solution of $\theta_{\bar E}
= S' ({\bar E})=1$. Applying $P'$ on both sides of this equation, we
deduce that ${\bar E} =P' (1)$. 

We claim that $J(\bar E)=0$. To prove this identity we need to show
that $S (\bar E)=\bar E$ or, in view of \eqref{S2}, that $\theta_{\bar
  E} \bar E - P(\theta_{\bar E})= \bar E$. Since $\theta_{\bar E}=1$,
this equation is reduced to $P(1)=0$, which is easy to check in view
of the explicit formula \eqref{f01} for the nonequilibrium pressure.

The same argument applies to ${\tilde J}$ and the result follows from
the identity
\begin{equation*}
P'(1)\;=\; {\tilde P}'(1) \;=\;
-\int_0^1 s( \alpha + (\beta -\alpha) x) \, dx\;.
\end{equation*}
\end{proof}

In \cite[Section 7]{DLS}, the authors compute the limit of the
variance of the sequences $(Y_N:N\ge 1)$ and $(\tilde Y_N:N\ge 1)$ and
show that the limits differ. This result can be recovered from a
second order expansion of the entropy function $S_{\alpha, \beta}$.

We have seen that the rate function $J$ has a unique minimum at ${\bar
  E}$. It is well known from the theory of large deviations that the
asymptotic variance of the sequence $Y_N$ is given by $J'' (\bar
E)^{-1} = - S'' ({\bar E})^{-1}$. Since $\theta_{\bar E}=1$ and since
$S$ is the Legendre transform of the nonequilibrium pressure $P$, we
have that $S'' ({\bar E})=1/P'' (\theta_{\bar E})=1/P'' (1)$. Hence,
$-P'' (1)$ is the asymptotic variance of the sequence $Y_N$.

By taking $\theta=1$ in (\ref{eq:PP}) we obtain that 
\begin{equation}
\label{eq:bv}
\begin{split}
- P'' (1) \;& =\; \cfrac{1}{\beta -\alpha} \int_{\alpha}^{\beta} 
t(1-t) \Big[ \log \Big( \cfrac{t}{1-t} \Big)\Big]^2 dt \\
& -\cfrac{1}{2(\beta -\alpha)^2} \int_\alpha^\beta dx \, 
\int_{\alpha}^{\beta} \, dy \, (s(x)-s(y))^2 \;.
\end{split}
\end{equation}
A long and tedious computation shows that this expression coincides
with the limiting variance derived in \cite{DLS}.

A similar computation in the equilibrium model gives that the
asymptotic variance of the sequence $\tilde Y_N$ is equal to
\begin{equation*}
\cfrac{1}{\beta -\alpha} \int_{\alpha}^{\beta} 
t(1-t) \Big[ \log \Big( \cfrac{t}{1-t} \Big)\Big]^2 dt\;.
\end{equation*}
In particular, the asymptotic variance in the nonequilibrium model is
strictly bounded above by the asymptotic variance in the equilibrium
model.

% \begin{figure}
% \label{fig:SS}
% \begin{center}
% \includegraphics[width=13cm]{Final_entropies.png}
% \caption{Graph of $S_{\alpha,\beta}(E)$ and $\tilde S_{\alpha,\beta}
% (E)$ w.r.t. $E$ for $\alpha=1/5, \beta=5/6$} 
% \end{center}
% \end{figure}

\appendix
\section{The assumption ({\bf H})}
\label{sec:H}

The stationary state $\mu^N_{\alpha, \beta}$ of the symmetric simple
exclusion process with open
boundary conditions can be expressed in terms of a product of matrices
\cite{DLS}: There exists matrices $D$, $E$ and vectors $|V\rangle,
\langle W|$ such that
\begin{equation*}
\begin{split}
& DE \;-\; ED \;=\;  D\;+\; E \; , \qquad
\{(1-\beta)D -\beta E\} \,|\, V \rangle \;=\; |\, V \rangle\;, \\
&\quad \langle W \,|\, \{\alpha E- (1-\alpha)D\} \;=\; \langle W \, |
\end{split}
\end{equation*}
and 
\begin{equation}
\label{a01}
\mu_{\alpha,\beta}^{N}(\eta) \;=\; \cfrac{\omega_N (\eta)}
{\langle W | (D+E)^{N-1} |V \rangle}\;,
\end{equation}
where the weight $\omega_N (\eta)$ is given by
\begin{equation*}
\omega_N (\eta) \;=\; \langle W | \prod_{x=1}^{N-1} 
\left\{ \eta(x) D + [1-\eta(x)] E\right\} | V \rangle\;.
\end{equation*}
The partition function $\langle W | (D+E)^{N-1} | V \rangle$ can be
computed explicitly \cite[(3.11)]{DLS}:
\begin{equation}
\label{a03}
\langle W | (D+E)^{N-1} | V \rangle \;=\;
\frac{(N+1)!}{2(\beta-\alpha)^N}\; \cdot
\end{equation}

Decompose the set $\{1, \ldots, N-1\}$ into $r=\ve^{-1}$ adjacent
intervals $K_1, \ldots,K_r$ of size $\ve N$ and denote by $\bM=(M_1,
\ldots, M_r)$ the number of particles in each box. We recall that
$\mu^N_{\alpha, \beta}(\,\cdot\, | \bM)$ is the probability measure
$\mu^N_{\alpha, \beta}$ conditioned to have $M_j$ particles in $K_j$,
$j=1, \ldots,r$.

Let $\eta^{+}$, $\eta^-$ be the configuration in $\Omega_N (\bM)$ with
all particles in each interval $K_j$ at the left most, right most
positions, respectively. Hence, if $K_j= \{x_1, \dots, x_L\}$,
$M_j=M$, $\eta^+(x) =1$ if and only if $x_1\le x\le x_M$, $\eta^-(x)=1$ 
if and only if
$x_{L-M+1} \le x\le x_L$.

\begin{lemma}
\label{s09}
For $0< \alpha\le \beta <1$, $\eta\in \Omega_N (\bM)$,
\begin{equation*}
\mu(\eta^- | \bM) \;\le\; \mu (\eta | \bM) \;\le\; \mu (\eta^+ |
\bM)\;. 
\end{equation*} 
\end{lemma}

\begin{proof}
This is a simple consequence of the matrix product form of the
stationary state. Let $\eta$ be any configuration in $\Omega_N$ and
let $1\le x\le N-1$ be any site such that $\eta (x)=1, \eta
(x+1)=0$. Then, $\omega_N (\sigma^{x,x+1}\eta) - \omega_N (\eta)$ is
equal to 
\begin{equation*}
\begin{split}
& \langle W | \prod_{y=1}^{x-1} \left\{ \eta(y) D + [1-\eta(y)] E\right\} 
\left[ ED - DE \right] \prod_{y=x+2}^{N-1} \left\{ \eta(y) D + [1-\eta(y)]
  E\right\} | V \rangle\\
&\quad =\; - \, \langle W | \prod_{y=1}^{x-1} \left\{ \eta(y) D + 
[1-\eta(y)] E\right\} \left[ D+E\right] \prod_{y=x+2}^{N-1} 
\left\{ \eta(y) D + [1-\eta(y)] E\right\} | V \rangle \\
&\quad =\; - \; \omega_{N-1} (\xi) \;-\; 
\omega_{N-1} (\zeta)  \; \le \; 0\;,  
\end{split}
\end{equation*}
where $\xi$, $\zeta$ are the configuration of $\Omega_{N-1}$ given by
$\xi = (\eta(1), \ldots, \eta(x-1), 1,\eta(x+2),\ldots, \eta(N-1))$,
$\zeta = (\eta(1), \ldots, \eta(x-1), 0,\eta(x+2),\ldots, \eta(N-1))$.
\end{proof}

Hence, the derivation of the assumption ({\bf H}) is reduced to the
proof that
\begin{equation*}
\lim_{\ve \to 0} \limsup_{N \to \infty} \sup_{\bM}  \cfrac{1}{N} \,
\Big\vert \log \mu^N_{\alpha, \beta} (\eta^+ | \bM) - 
\log \mu^N_{\alpha, \beta} (\eta^- | \bM)\Big\vert
\;=\; 0\;. 
\end{equation*}
or, equivalently, to the proof that
\begin{equation*}
\lim_{\ve \to 0} \limsup_{N \to \infty} \sup_{\bM}  \cfrac{1}{N} \,
\Big\vert \log \mu^N_{\alpha, \beta} (\eta^+) - 
\log \mu^N_{\alpha, \beta} (\eta^-)\Big\vert
\;=\; 0\;. 
\end{equation*}

For each fixed $\ve >0$ and $\bM$, the configurations $\eta^+$,
$\eta^-$ are associated to density profiles $\rho^+$,
$\rho^-:[0,1]\to[0,1]$ defined by $\rho^+ = \sum_{1\le i \le r}
1\{[(i-1)\varepsilon , (i-1)\varepsilon + \rho_i)\}$, $\rho^- =
\sum_{1\le i \le r} 1\{[i\varepsilon - \rho_i , i\varepsilon )\}$,
where $\rho_i = M_i/N$, $1\le i\le r$ and $1\{A\}$ stands for the
indicator of the set $A$. Therefore, by the large deviation principle
for the density profiles under the stationary state $\mu^N_{\alpha,
  \beta}$ \cite{DLS, BDGJL2, bg,f},
\begin{equation*}
\lim_{N \to \infty} \cfrac{1}{N} \, \log \mu^N_{\alpha, \beta}
(\eta^\pm) \;=\; - V_{\alpha, \beta}(\rho^\pm)\; ,
\end{equation*}
where $V_{\alpha, \beta}$ is the functional introduced just before
\eqref{eq:bvp}. 

Hence, to prove assumption ({\bf H}), it remains to show that
\begin{equation*}
\lim_{\ve \to 0} \sup_{\rho}  \,
\Big\vert  V_{\alpha, \beta} (\rho^+) -  V_{\alpha, \beta} (\rho^-)\Big\vert
\;=\; 0\;. 
\end{equation*}
where the supremum is carried over all $0\le \rho_i\le \varepsilon$,
$1\le i\le r$. Since $\rho^\pm(x)$ is either $0$ or $1$, we may
replace in the previous formula, $V_{\alpha, \beta}$ by $V_{\alpha,
  \beta} + \bb S$. By \cite[Theorem 4.1]{BSGJL}, this functional is
continuous in $\mc M$. For each $\varepsilon>0$, denote by $\rho^{\pm,
  \varepsilon}$ the profiles which attain the previous supremum with
$V_{\alpha, \beta}$ by $V_{\alpha, \beta} + \bb S$. By compactness of
$\mc M$, there exists a subsequence $\varepsilon_k \downarrow 0$ for
which
\begin{equation*}
\begin{split}
& \lim_{\ve \to 0} \sup_{\rho}  \, \Big\vert  V_{\alpha, \beta} (\rho^+) 
-  V_{\alpha, \beta} (\rho^-)\Big\vert \\
& \quad =\; \lim_{k \to \infty}   \, \Big\vert  V_{\alpha, \beta} 
(\rho^{+, \varepsilon_k}) + \bb S (\rho^{+, \varepsilon_k})
-  V_{\alpha, \beta} (\rho^{-, \varepsilon_k}) - \bb S (\rho^{-,
  \varepsilon_k}) \Big\vert   
\end{split}
\end{equation*}
and $\rho^{+, \varepsilon_k}$ converges weakly to some profile
$\rho$. Clearly, the sequence $\rho^{-, \varepsilon_k}$ converges
weakly to the same profile $\rho$. Since $V_{\alpha, \beta} + \bb S$
is continuous in $\mc M$, assumption ({\bf H}) is proved.

\medskip
\noindent{\bf Acknowledgments.} The authors thank B. Derrida for
indicating to them Lemma \ref{s09}, L. Bertini, D. Gabrielli, G.
Jona-Lasinio and J. L. Lebowitz for fruitful discussions and the
referees for their remarks. The end of Section 7, in particular, was
added by suggestion of one of the referees. The warm hospitality of
IMPA is also acknowledged by the first author.


\begin{thebibliography}{10}

\bibitem{B} Bahadoran, C., On the convergence of entropy for
  stationary exclusion processes with open boundaries, J. Stat. Phys.
  {\bf 126} (2007), no. 4-5, 1069--1082.

\bibitem{BDGJL2} Bertini L., De Sole A., Gabrielli D., Jona-Lasinio
  G., Landim C., \emph{Macroscopic fluctuation theory for stationary
    non equilibrium state.}  J.\ Statist.\ Phys.\ {\bf 107}, 635--675
  (2002).

\bibitem{BSGJL} L. Bertini, A. De Sole, D. Gabrielli, G. Jona-Lasinio,
  C. Landim, Large deviations for the boundary driven symmetric simple
  exclusion process, Math. Phys. Anal. Geom. {\bf 6} (2003), no. 3,
  231--267.

\bibitem{BDGJL9} Bertini L., De Sole A., Gabrielli D., Jona-Lasinio
  G., Landim C., \emph{Stochastic interacting particle systems out of
    equilibrium.}  J. Stat. Mech. P07014. (2007)

\bibitem{BSGJL2} Bertini, L.; De Sole, A.; Gabrielli, D.;
  Jona-Lasinio, G.; Landim, C. Towards a nonequilibrium
  thermodynamics: a self-contained macroscopic description of driven
  diffusive systems. J. Stat. Phys. {\bf 135} (2009), no. 5-6,
  857--872.

\bibitem{bdgjl13} Bertini, L.; De Sole, A.; Gabrielli, D.;
  Jona-Lasinio, G.; Landim, C. Action functional and quasi-potential
  for the Burgers equation in a bounded interval. arXiv:1004.2225
  (2010).

\bibitem{bg} Bodineau T., Giacomin G., \emph{From dynamic to static
    large deviations in boundary driven exclusion particles systems.}
  Stoch. Proc. Appl. {\bf 110}, 67--81 (2004).

\bibitem{D} Derrida, B., Non-equilibrium steady states: fluctuations
  and large deviations of the density and of the current,
  J. Stat. Mech. P07023 (2007)

\bibitem{DLS} Derrida,B., Lebowitz, J.L., Speer, E.R., Large Deviation
  of the Density Profile in the Steady State of the Open Symmetric
  Simple Exclusion, J. Statist. Phys. {\bf 107} (2002), no. 3-4,
  599--634.

\bibitem{DLS2} Derrida, B.; Lebowitz, J. L.; Speer, E. R.; Entropy of
  open lattice systems, J. Stat. Phys. {\bf 126} (2007), no. 4-5,
  1083--1108.

\bibitem{f} Farfan J.; \emph{Static large deviations of boundary
    driven exclusion processes.}  Preprint 2009,
  \texttt{arXiv:0908.1798v1}

\bibitem{P} Presutti, E., Scaling limits in statistical mechanics and
  microstructures in continuum mechanics. Theoretical and Mathematical
  Physics. Springer, Berlin, 2009. xvi+467 pp. ISBN: 978-3-540-73304-1

\bibitem{S} H. Spohn, {\sl Long range correlations for stochastic
    lattice gases in a nonequilibrium steady state.}\/ J. Phys. A {\bf
    16} (1983), 4275--4291.

\end{thebibliography}
\end{document}